\documentclass[10pt, a4paper]{amsart}

\usepackage{xcolor}
\definecolor{MyBlue}{cmyk}{1,0.13,0,0.63}
\definecolor{MyGreen}{cmyk}{0.91,0,0.88,0.52}
\newcommand{\mylinkcolor}{MyBlue}
\newcommand{\mycitecolor}{MyGreen}
\newcommand{\myurlcolor}{black}

\usepackage{hyperref}
\hypersetup{%
  bookmarksnumbered=true,bookmarksopen=false,%
  plainpages=false,
  linktocpage=true,%
  colorlinks=true,breaklinks=true,%
  linkcolor=\mylinkcolor,citecolor=\mycitecolor,urlcolor=\myurlcolor,%
  pdfpagelayout=OneColumn,%
  pageanchor=true,%
}

\usepackage{graphicx}
\usepackage{amsmath, amsthm,  amsfonts, amscd}
\usepackage{amssymb}
\usepackage{url}
\usepackage{fullpage}
\usepackage{mathtools}
\usepackage[all]{xy}
\usepackage{slashed}
\usepackage{multirow}

\newtheorem{thm}{Theorem}[section]
\newtheorem*{thm*}{Theorem}
\newtheorem{cor}[thm]{Corollary}
\newtheorem{lemma}[thm]{Lemma}
\newtheorem{prop}[thm]{Proposition}

\theoremstyle{definition}
\newtheorem{defn}[thm]{Definition}
\newtheorem{assumption}[thm]{Assumption}
\theoremstyle{remark}
\newtheorem{remark}[thm]{Remark}

\newtheorem{example}[thm]{Example}

\numberwithin{equation}{section}

\newcommand{\End}{\ensuremath{\mathrm{End}}}

\newcommand{\wh}{\ensuremath{\widehat}}
\newcommand{\wt}{\ensuremath{\widetilde}}

\newcommand{\R}{\ensuremath{\mathbb{R}}}
\newcommand{\N}{\ensuremath{\mathbb{N}}}
\newcommand{\Z}{\ensuremath{\mathbb{Z}}}

\newcommand{\C}{\ensuremath{\mathbb{C}}}
\newcommand{\T}{\ensuremath{\mathbb{T}}}

\def\calT{\mathcal{T}}
\def\calC{\mathcal{C}}

\def\calK{\mathcal{K}}
\def\calB{\mathcal{B}}
\def\calH{\mathcal{H}}
\def\calD{\mathcal{D}}

\def\calA{\mathcal{A}}

\def\calE{\mathcal{E}}

\newcommand{\ol}{\overline}

\theoremstyle{definition}

\DeclareMathOperator{\Dom}{Dom}

\DeclareMathOperator{\Index}{Index}

\DeclareMathOperator{\Ker}{Ker}

\newcommand{\rst}[1]{\ensuremath{{\mathbin\upharpoonright}%
\raise-.5ex\hbox{$#1$}}}

\makeatletter

\newcommand{\Rmnum}[1]{\expandafter\@slowromancap\romannumeral #1@}
\makeatother

\author{C. Bourne}
\address{Mathematical Sciences Institute, Australian National University, Canberra, ACT 0200, Australia. School of Mathematics and Applied Statistics, University of Wollongong, Wollongong, NSW 2522, Australia. 
Department Mathematik, Friedrich-Alexander-Universit\"{a}t Erlangen-N\"{u}rnberg, 
Cauerstra{\ss}e 11, 91058 Erlangen, Germany}
\email{bourne@math.fau.de}

\author{A. L. Carey}
\address{Mathematical Sciences Institute, Australian National University, Canberra, ACT 0200, Australia. School of Mathematics and Applied Statistics, University of Wollongong, Wollongong, NSW 2522, Australia}
\email{alan.carey@anu.edu.au}

\author{A. Rennie}
\address{School of Mathematics and Applied Statistics, University of Wollongong, Wollongong, NSW 2522, Australia}
\email{renniea@uow.edu.au}

\date{\today}

\begin{document}

\begin{abstract}
We study topological insulators, regarded as physical systems giving rise to 
topological invariants determined by symmetries both linear and anti-linear. 
Our perspective is that of noncommutative index theory of operator algebras. 
In particular we formulate the index problems using Kasparov theory, both 
complex and real. We show that the periodic table of topological insulators 
and superconductors can be realised as a real or complex index pairing of a 
Kasparov module capturing internal symmetries of the Hamiltonian with a 
spectral triple encoding the geometry of the sample's (possibly noncommutative) Brillouin zone. 
\end{abstract}

\title{A noncommutative framework for topological insulators}
\maketitle



\tableofcontents


\section{Introduction}
\subsection{Overview}
This paper is an investigation of the links between topological states of matter 
and real Kasparov theory.
It is independent of (though motivated by) our previous paper~\cite{BCR14}, 
where we examined a discrete (tight-binding) 
model of the integer quantum Hall effect in complex Kasparov theory, which we briefly describe. 

The discrete model of the quantum Hall effect gives rise to a 
spectral triple in the boundary-free (bulk) system (we explain this terminology later). 
When a boundary/edge is added to the model, bulk and edge systems can be 
linked by a short exact sequence~\cite{SBKR02}. In \cite{BCR14} we proved that the bulk 
spectral triple can be factorised as the Kasparov product of a 
Kasparov module representing the short exact sequence with a spectral triple 
coming from observables concentrated at the boundary of the sample.

In~\cite{BCR14}, we commented that the method used allowed the 
Hall conductance to be computed without passage to cyclic homology and 
cohomology as in~\cite{KR06,SBKR02,KSB04b}, and that this would be 
useful for detecting torsion invariants that cyclic theory cannot detect.

In this paper, we show that Kasparov theory in the real case (that is $KKO$-theory) 
may be used to give a noncommutative framework for topological insulators in 
which torsion invariants may be detected and effectively computed. 
Our invariants are (real) Kasparov classes, and so automatically  topological 
invariants. These Kasparov modules can be naturally identified
with Clifford modules, and then with $KO$-classes (of a point) via the 
treatment of $KO$-theory
by Atiyah-Bott-Shapiro \cite{ABS64}.
This procedure is in contrast to the more usual methods of obtaining
$\Z_2$-invariants in the literature. Eigenvalue counting (mod 2) and
similar methods do not a priori produce topological invariants, necessitating an extra step to 
prove homotopy invariance.

One motivation for writing this article as a review is to give some 
background on the application of Kasparov theory to topological states 
of matter. The interested reader can go further into the study of index theory 
applicable to systems that require real $K$-theory and noncommuative methods.  
At the level of the algebra of observables, systems with anti-linear symmetries 
require the algebra to be real or Real (note the capitalisation). 
Elsewhere we will see that the $K$-theoretic formulation of index theory 
provides a framework to consider disordered systems and the bulk-edge 
correspondence for topological insulators\footnote{In work in progress \cite{BKR} we 
have explored these issues in some detail.}.

The `periodic table for topological insulators' has a long history.  We refer 
the reader to the work of Altland and Zirnbauer \cite{AZ}, Ryu et al.~\cite{RSFL10,SRFL08} and Kitaev \cite{Kitaev09}.
The idea of these and other authors is to show how the bulk invariants of interest 
in systems with symmetries, such as topological insulators, can be related to real 
and complex (commutative) $K$-theory. More recently, papers of a more 
mathematical nature have appeared establishing the $K$-theoretic properties 
of topological insulators~\cite{DNSB14b,FM13,GSB15,Kellendonk15,Thiang14}. 

The approach of these papers is adapted to the study of the (commutative) 
geometry of the Brillouin zone. It is widely appreciated that the introduction 
of disorder will require a
noncommutative approach. One purpose of this article is to propose a 
methodology flexible enough to handle noncommutative observable algebras. 
Whilst not completely addressing the issue of disorder, we do show that our technique extends to the
noncommutative algebra of disordered Hamiltonians provided that they
retain a spectral gap.
 
We remark that even in the commutative viewpoint the links between the 
bulk-edge correspondence, real/Real $K$-theory and $K$-homology 
are less studied.  In this article we examine how real Kasparov theory can 
be used to derive the invariants of interest for bulk systems (i.e. systems with no boundary).  
In a separate work we will show 
how the framework developed here is employed to establish a bulk-edge correspondence 
of topological insulators in arbitrary dimension.

The novelty of our approach is that it exploits the full bivariant $KK$-theory 
as developed by Kasparov~\cite{Kasparov80}, and utilises the unbounded setting 
as developed in~\cite{BJ83,BMvS13,FR15,KL13,Kucerovsky97,MeslandMonster, MR15}
so that all our constructions are explicit and   have natural 
physical and/or geometric interpretations.

\subsection{Outline of the paper}
The new results in this article begin in Section \ref{sec:bulk_theory} after a 
review in Section \ref{sec:Insulator_lit_review} of previous work.  These new results concern 
 a derivation of the $K$-theoretic classification of topological states of matter in Kasparov theory, 
 complex and real. The work  in~\cite{GSB15,Kellendonk15,Thiang14} also translates into this language. 

While a derivation of the periodic table is not exactly new, both in the physics 
and mathematical literature, we are of the opinion that an understanding of 
the systematic application of Kasparov's powerful machinery to the insulator problem is 
critical for further developments, such as `index theorems' 
for torsion invariants and the incorporation
of disorder.  (Note that the approach to index theorems for insulators using 
the non-commutative Chern character, \cite{PLB13,PSB14}, is not applicable 
to torsion invariants. See \cite{GSB15} for an alternative viewpoint to ours.) 
In the second half of Section \ref{sec:bulk_theory} we show how invariants 
arise through a Clifford module valued index using the Kasparov product.

In the last section we show how our general method applies to some of the 
examples of interest in the physics and mathematics literature. This includes 
the well-known Kane-Mele model of time-reversal invariant systems as well as 
$3$-dimensional examples.

For the reader's benefit, we also include an appendix that compiles some of the 
basic definitions and
results of Kasparov theory for real $C^*$-algebras.

The second stage in this program is to study the  bulk-edge correspondence of 
topological insulators using Kasparov theory. We make some preliminary comments in Section \ref{subsec:bulk_edge}, 
but delay a full investigation for future work~\cite{BKR}. The use of Kasparov 
theory and the intersection product to study systems with boundary can potentially 
be extended further to systems with disorder and continuous models.

\subsubsection*{Acknowledgements}
The authors thank Hermann Schulz-Baldes 
for useful discussions. We also thank Johannes Kellendonk, Guo Chuan Thiang 
and Yosuke Kubota for a careful reading of an 
earlier version of this paper. All authors thank the Hausdorff Institute for Mathematics for support. 
CB and AC thank the Advanced Institute for Materials Research at Tohoku University for 
hospitality while some of this article was written. All authors acknowledge the support 
of the Australian Research Council.


\section{Review}
\label{sec:Insulator_lit_review}

\subsection{Motivating example: the integer quantum Hall effect} \label{subsec:IQHE}

Our framework for topological insulators builds from the use of 
noncommutative geometry to explain the quantum Hall effect by 
Bellissard and others (see~\cite{Bellissard94, NB90}), 
which uses the language of Fredholm modules. 
In our previous work \cite{BCR14} we preferred the language of 
spectral triples because of the close relationship with (unbounded) Kasparov
theory which formed an essential tool in our approach.  
Recall that a complex spectral triple $(\mathcal A, \mathcal H, \mathcal D)$ consists of
an algebra $\mathcal A$ of (even) operators on a $\Z_2$-graded separable complex Hilbert space 
$\mathcal H$ and a distinguished densely defined (odd) self-adjoint operator $\mathcal D$
(herein referred to as a `Dirac-type' operator) with the properties that
commutators $[\mathcal D,a]$ are bounded for all $a\in \mathcal A$ 
and products $a(1+\mathcal D^2)^{-1/2}$ are compact for all $a\in \mathcal A$.

\subsubsection{The discrete integer quantum Hall system}
We review   the discrete or tight-binding quantum Hall system 
as considered in~\cite{Bellissard94,SBKR02, CM96}. This model motivates our later discussions 
and allows our constructions and computations to be as transparent as possible. 

In the case without boundary, where $\calH =\ell^2(\Z^2)$, we have magnetic translations 
$\wh{U}$ and $\wh{V}$ acting as unitary operators on $\ell^2(\Z^2)$. These operators 
commute with the unitaries $U$ and $V$ that generate the Hamiltonian $H=U+U^* + V + V^*$. 
We choose the Landau gauge so that 
\begin{align*}
   &(U\lambda)(m,n) = \lambda(m-1,n),  &&(V\lambda)(m,n)=e^{-2\pi i\phi m}\lambda(m,n-1), \\
   &(\wh{U}\lambda)(m,n) = e^{-2\pi i \phi n}\lambda(m-1,n),  &&(\wh{V}\lambda)(m,n) = \lambda(m,n-1),
\end{align*}
where $\phi$ has the interpretation as the magnetic flux through a unit cell and 
$\lambda\in\ell^2(\Z^2)$. We note that $\wh{U}\wh{V}=e^{-2\pi i\phi}\wh{V}\wh{U}$ 
and $UV = e^{2\pi i\phi}VU$, so $C^*(U,V)\cong A_{\phi}$, the rotation algebra, 
and $C^*(\wh{U},\wh{V})\cong A_{-\phi}$. We can also interpret 
$A_{-\phi} \cong A_{\phi}^{\mathrm{op}}$, where $A_\phi^\mathrm{op}$ 
is the opposite algebra with multiplication 
$(ab)^\mathrm{op} = b^\mathrm{op} a^\mathrm{op}$. To see this identification we compute,
$$ 
{U}^\mathrm{op}{V}^\mathrm{op} 
= \left({V}{U}\right)^\mathrm{op} 
= e^{-2\pi i\phi}\left({U}{V}\right)^\mathrm{op} 
= e^{-2\pi i\phi}{V}^\mathrm{op}{U}^\mathrm{op}. 
$$

Our choice of gauge also means that $C^*({U},{V})\cong C^*({U})\rtimes_\eta \Z$, 
where ${V}$ is implementing the crossed-product structure via the automorphism 
$\eta({U}^m) = {V}^* {U}^m {V}$.

\medskip

The topological properties of the quantum Hall effect 
come from two pieces of information: the Fermi projection 
of the Hamiltonian and a spectral triple encoding the 
geometry of the (noncommutative) Brillouin zone. 
Provided the Fermi level $\mu\notin \sigma(H)$, 
then the Fermi projection $P_\mu$ defines a class 
in the complex $K$-theory group $[P_\mu]\in K_0(A_{\phi})$.

\begin{prop}[\cite{BCR14}] 
\label{prop:qH_spec_trip}
Let $\calA_\phi$ be the dense $\ast$-subalgebra of 
$A_\phi$ generated by finite polynomials of $U$ and $V$, and let $X_1,\,X_2$ be
the operators on $\ell^2(\Z^2)$ given by $(X_j\xi)(n_1,n_2)=n_j\xi(n_1,n_2)$. Then
$$ 
\left( \calA_{\phi},\, \ell^2(\Z^2)\otimes\C^2=\begin{pmatrix} \ell^2(\Z^2)\\ \ell^2(\Z^2)\end{pmatrix}, \,
\begin{pmatrix} 0 & X_1-iX_2 \\ X_1+iX_2 & 0 \end{pmatrix}, \,
\gamma=\begin{pmatrix} 1 & 0 \\ 0 & -1 \end{pmatrix} \right) 
$$
is a complex spectral triple.
\end{prop}

We consider the index pairing of $K$-theory and $K$-homology,
\begin{align*}
   K_0(A_{\phi}) \times K^0(A_{\phi}) &\to K_0(\C)\cong \Z \\
    \left([P_\mu], [X] \right) &\mapsto \Index(P_\mu (X_1+iX_2) P_\mu),
\end{align*}
where $[X]$ is the $K$-homology class of the spectral triple from 
Proposition \ref{prop:qH_spec_trip}.
A key result of Bellissard is that the Kubo formula for the 
Hall conductance can be expressed in terms of this index pairing~\cite{Bellissard94}. Namely,
\begin{equation} 
\label{eq:Bellissard_Hall_conductance}
 \sigma_H = \frac{2\pi i e^2}{h}\calT( P_\mu [\partial_1 P_\mu, \partial_2 P_\mu] ) 
 = \frac{e^2}{h} \Index(P_\mu (X_1+iX_2) P_\mu), 
\end{equation}
where $\calT$ is the trace per unit volume and $\partial_j(a) = -i[X_j,a]$ for $j=1,2$. 

The tracial formula for the conductance, also called the 
Chern number of the projection $P_\mu$, gives a computationally tractable 
expression for the index pairing. The expression comes from 
translating the index pairing of $K$-theory and $K$-homology 
into a pairing of cyclic homology and cyclic cohomology. This can be 
done for integer invariants but can \emph{not} be done for 
torsion invariants, e.g. the $\Z_2$-invariant associated to the 
time-reversal invariant systems. Therefore a cyclic formula does not 
arise in general topological insulator systems and instead 
we must deal with the $K$-theoretic index pairing directly.

Let us also briefly review the bulk-edge correspondence for the quantum Hall 
effect as studied by Kellendonk, Richter and Schulz-Baldes~\cite{SBKR00,SBKR02, KSB04a,KSB04b}. 
The algebra $A_\phi$ acts on a system without boundary and is 
considered as our `bulk algebra'. Kellendonk et al. consider a system 
with boundary/edge and construct a short exact sequence of $C^*$-algebras
$$  
0 \to C^*(U)\otimes\calK[\ell^2(\N)] \to \calT \to A_\phi \to 0, 
$$
where $C^*(U)$ acts on $\ell^2(\Z)$ and $\calT$ is a Toeplitz-like 
extension~\cite{PV80}. We consider elements in 
$C^*(U)\otimes\calK[\ell^2(\N)]$ as observables concentrated at the 
boundary of our system with edge. Kellendonk et al. define an 
`edge conductance' from the algebra $C^*(U)$, a topological invariant, 
and show that under suitable hypothesis this invariant is the same 
as the Hall conductance associated to $A_\phi$. Kellendonk et al. 
prove this result by considering the Hall conductance (and edge conductance) 
as a pairing between $K$-theory and cyclic cohomology. Recently, 
the authors adapted Kellendonk et al.'s general argument to Kasparov 
theory, avoiding the passage to cyclic cohomology~\cite{BCR14}. In 
future work, we will refine this argument to deal with real algebras 
and torsion invariants, which is required in topological insulator 
systems with time-reversal or charge-conjugation symmetry~\cite{BKR} 
(a brief summary is given in Section \ref{subsec:bulk_edge}).

\subsubsection{Higher-order Chern numbers}
We may also consider higher-dimensional magnetic Hamiltonians 
as elements in $A_{\phi}^d$, the noncommutative $d$-torus. 
We take the dense $\ast$-subalgebra $\calA_\phi^d \subset A_\phi^d$ 
of finite polynomials of (twisted) shift operators. 
As studied in~\cite{PLB13,PSB14}, there is a natural 
extension of our quantum Hall spectral triple to arbitrary dimension. 
\begin{prop}
The tuple
\begin{equation*} 
   \left( \calA_\phi^d,\, \ell^2(\Z^d)\otimes \C^N\otimes \C^\nu,\, \sum_{j=1}^d X_j\otimes 1_N\otimes \gamma^j \right)
\end{equation*}
is a spectral triple, where the matrices $\gamma^j\in M_\nu(\C)$ 
satisfy $\gamma^i \gamma^j + \gamma^j\gamma^i = 2\delta_{i,j}$ and so generate 
the complex Clifford algebra $\C\ell_d$. If $d$ is even, the spectral triple 
is graded by the operator $\gamma = (-i)^{d/2}\gamma^1\cdots\gamma^d$.
\end{prop}
 We obtain `higher order Chern numbers' by taking the index 
 pairing of this spectral triple with the Fermi projection or 
 some unitary $u\in\calA$ (see also~\cite{PLB13,PSB14}).

To put this in the language of Kasparov theory we need the Kasparov product.  
In general this product
requires three $C^*$-algebras $A,B,C$ and it gives a composition rule for $KK$-classes
of the form $$KK(A,B)\times KK(B,C) \to KK(A,C).$$
We will not attempt to review the construction of the Kasparov product here 
referring to \cite{Blackadar,Kasparov80} for details.
In this paper we will exploit explicit constructions of the product at the 
level of (unbounded) representatives of Kasparov classes.
For these constructions it is not necessary to understand Kasparov's 
approach in detail, but to utilise
the work of Kucerovsky~\cite{Kucerovsky97} in simplifying the 
issues in geometric/physical examples.

 For $d$ even the particular Kasparov product relevant to our discussion will be  the following:
\begin{align*}
     KK(\C, A) \times KK(A,\C) \to KK(\C,\C) \cong \Z 
\end{align*}
The relevance of this product is due to the fact that  $KK(\C, A)$ 
is isomorphic to the $K$-theory of the algebra $A$
while elements of  $KK(A,\C)$ are represented by spectral triples 
and the integer constructed by the pairing in this instance
is, in fact, a Fredholm index.
Applying this  language of Kasparov theory to the even-dimensional systems, we have
the product
\begin{align*}
     K&K(\C, A_\phi^d) \times KK(A_\phi^d,\C) \to KK(\C,\C) \cong \Z 
\end{align*}
and  the integer constructed from this product is given by
\begin{align*}
    C_d &:= [P_\mu] \hat\otimes_A \left[\left( \calA_\phi^d, \ell^2(\Z^d)\otimes \C^N\otimes \C^\nu, X=\sum_{j=1}^d X_j\otimes 1_N\otimes \gamma^j, \gamma \right) \right] \\
       &= \Index(P_\mu X_+ P_\mu),
\end{align*}
where $X = \begin{pmatrix} 0 & X_- \\ X_+ & 0 \end{pmatrix}$ is 
decomposed by representing the $\Z_2$-grading 
$\gamma = \begin{pmatrix} 1 & 0 \\ 0 & -1\end{pmatrix}$. The case of $d$ odd 
has an analogous formula but we are taking a product of 
$KK(\C\ell_1,A_\phi^d)$ with $KK(A_\phi^d,\C\ell_1)$. 

In this paper we will  replace this complex index pairing by its real analogue.
We remark however that in the real picture, when one considers 
time-reversal and charge-conjugation symmetry,
the index given by $KKO$-theory is naturally Clifford module 
valued and the identification with elements of
$\Z$ or $\Z_2$ is an additional step.

In this paper we also use real spectral triples.
The definition is the same with real replacing complex everywhere. 
These will enter our discussion in Section \ref{subsec:spec_trip_and_pairings}.
 We begin, however, with some background on symmetries.

\subsection{Symmetries and invariants}
Topological insulators can be loosely described as physical systems 
possessing certain symmetries which give rise to invariants topologically 
protected by these symmetries. The symmetries of most interest to 
physicists are time-reversal symmetry, charge-conjugation symmetry 
(also called particle-hole symmetry) and sublattice symmetry 
(also called chiral symmetry). We do, however, note that other 
symmetries such as spatial inversion symmetry may be 
considered though they will not play a central role here.

The integer quantum Hall effect motivates our approach as it is topological: the 
Hall conductance can be expressed in terms of a pairing of homology classes 
of certain bundles over the Brillouin zone (momentum space) of our sample. 
Of course, in order to properly understand the meaning of a bundle over the 
Brillouin zone in the case of irrational magnetic field, one has to pass to the 
noncommutative picture.  Because the quantised Hall conductance is a 
topological property, it is stable under small perturbations and the addition 
of impurities into the system. Indeed, disorder plays an important role 
in the localisation of electrons and the stable nature of the Hall 
conductance in between jumps~\cite{Bellissard94}. 

The integer quantum Hall effect is linked to topological 
information but does not possess any of the additional 
symmetries discussed above. Hence the effect can be considered 
as an example of a topological insulator without 
additional symmetry.

Much more recently, the prediction of a new topological state of matter 
came from Kane and Mele~\cite{KM05b}, who consider two copies of a 
Haldane system (that is, a single-particle Hamiltonian acting on a honeycomb lattice) 
and impose time-reversal symmetry on their model. By considering 
the topology of the spectral bands of the Hamiltonian 
under time-reversal symmetry (considered as 
bundles over the Brillouin zone), the authors associate a $\Z_2$ 
parameter to their system. This number is `topologically protected' 
because one cannot pass from one value of the parameter to the other 
unless time-reversal symmetry is broken. 

A non-trivial Kane-Mele invariant is predicted to 
be related to the existence of edge channels that carry opposite currents along the 
edge of a sample. These edge channels are said to be 
spin-oriented; that is, the spin-up and spin-down electrons separate and give currents 
travelling in opposite directions. The net current is zero, but each spin component 
has a non-trivial conductance that can be linked to topological 
invariants of classical bundles over the Brillouin zone. 
Such an effect is also called the quantum spin-Hall effect.

The quantum spin-Hall effect was initially predicted to occur in graphene, 
but this is difficult to verify experimentally. The effect was later 
predicted to be found in HgTe~\cite{BHZ06}, a compound much more 
amenable to experimental analysis, and subsequently the 
theoretical prediction was measured in~\cite{KWBRBMQZ07}.

At the time of writing, more recent experimental work has called into 
dispute the link between time reversal invariance and the existence 
of robust edge channels. Edge channels have been experimentally 
observed in quantum spin-Hall systems with an external magnetic field of up to 
9T~\cite{spinQH_with_mag_field}. Hence there is still work to be done 
to determine whether a $2$-dimensional time-reversal invariant system 
gives rise to stable edge currents.

Despite the current experimentation issues, 
the Kane-Mele invariant opened up a new avenue of theoretical research. 
The existence of a torsion invariant 
in this system has provoked substantial effort in understanding whether  
similar invariants of a finer type could be found in other models and systems. 
This included higher-dimensional time-reversal invariant insulators, 
experimentally found in~\cite{3d_verification}. Particle-hole/charge-conjugation 
symmetric systems were also considered, which drew a link to superconductors, 
whose current can be considered as the scattering of an 
electron by a hole (see for example~\cite{QHZ08, SRFL08}).

Many  possible models were quickly discovered and the question 
began to turn towards how to properly classify such systems 
from their symmetry data. This involved showing how 
the `topological numbers' derived in the various systems could be 
connected to algebraic topology, specifically classifying spaces 
and homotopy groups of symmetry compatible Hamiltonians. 
While there are many papers on this topic, one of the most 
influential came from Kitaev~\cite{Kitaev09}, who outlined how 
symmetry data can be linked to Clifford algebras and, in particular, 
$K$-theory. Specifically, if one considers a system with time-reversal, 
charge-conjugation or sublattice symmetry, then then one finds 
ten different outcomes depending on the nature of the 
symmetry (see~\cite{Kitaev09, RSFL10} for more on this). 
Kitaev argued that these different outcomes correspond 
precisely to the $10$ different $K$-theory groups ($8$ real groups 
and $2$ complex groups), where the $K$-theory is again coming 
from bundles over the Brillouin zone. The paper also showed 
how the dimension of the system affects the kind of invariant that may arise.

The work of Kitaev and Ryu et al. has been expanded and developed in newer 
papers by Stone et al.~\cite{SCR11} and Kennedy and Zirnbauer~\cite{KZ15}. 
To briefly summarise, Stone et al. and Kennedy-Zirnbauer 
are able to link the symmetries of interest to stable homotopy groups and Clifford algebras 
in a way that is more physically concrete than Kitaev's original outline. 
In particular, Kennedy and Zirnbauer 
show how the Bott periodicity of complex and real $K$-theory can be understood in terms of the 
symmetries of the system~\cite{KZ15}.

\subsection{Contributions in the mathematical literature}
While there has been a plethora  of articles in the physics 
literature about topological insulators and their properties, 
there are comparatively few mathematical physics papers 
on the subject (that is, papers with a mathematical focus, 
but with physical applications in mind). Despite this, there 
have been some important contributions from mathematical 
physicists in understanding the mechanics of insulator systems, 
particularly with regard to making Kitaev's $K$-theoretic 
classification more explicit. We briefly review the work 
that has been done on these problems, focusing on the more 
$K$-theoretic papers as these are closest to our viewpoint. 
We do not claim that our review is comprehensive or complete 
although from the work cited here a comprehensive bibliography could be built.
Our purpose is  to highlight what is understood and what remains open.

\subsubsection*{Almost commuting matrices and insulator systems (Loring et al.)}
Some of the first mathematical attempts to understand the 
topological insulator problem came from Loring in collaboration 
with Hastings and S{\o}rensen~\cite{HL10a, HL11, Loring15, HL10b, LS10, LS13, LS14}.

Very roughly speaking, these papers start with the model of a finite lattice 
on a torus or sphere. There are translation operators $U_i$ between 
atom sites that commute. However, when these operators are 
compressed by the Fermi projection $P_\mu U_i P_\mu$, 
they may no longer commute. The observation of the authors 
is that the act of approximating the matrices $P_\mu U_i P_\mu$ 
with commuting matrices can be viewed as a lifting problem in 
$C^*$-algebras. The authors then argue that the obstructions to 
approximating almost-commuting matrices with commuting matrices 
lead to $K$-theory invariants (both complex and real). 
These obstructions can then be related back to the various symmetries 
that arise in insulator systems.

The papers of Loring, Hastings and S{\o}rensen are able to 
mathematically establish a link between insulator systems and 
$K$-theory of operator algebras, though the physical models 
used (a finite lattice on a $d$-torus or $d$-sphere) do not line 
up easily with the models that are usually considered. Another 
drawback is that the methods Loring et al. use are quite different 
to any other treatment of such systems (including the various 
explanations of the integer quantum Hall effect) and so are difficult 
to adapt to the physical interpretations of such systems.

By considering the topological insulator problem and it's link to real/Real 
$K$-theory, there have also been some useful mathematical papers 
explaining $KKR$ and $KO$-theory~\cite{BLR12, BL15}. In particular, 
the paper~\cite{BL15} provides a helpful characterisation of all $8$ 
$KO$-groups in terms of unitary matrices and involutions.

\subsubsection*{Bloch bundles and $K$-theory (De Nittis-Gomi)}
An alternate viewpoint comes from the papers of De Nittis 
and Gomi~\cite{DNG14a, DNG14b, DNG15a, DNG15b}, 
who are developing a more explicitly geometric interpretation 
of the insulator invariants that arise. This is done by constructing 
a theory of Real or quaternionic or chiral vector bundles, 
and showing how the topological properties of insulator 
systems can be interpreted as geometric invariants of 
these bundles over the Brillouin zone. This work serves to 
correct some inconsistencies in the physics literature, where 
many of the bundles considered are trivial and symmetry structures 
implemented globally. De Nittis and Gomi show that when only 
local trivialisations are considered, much more care needs to be 
taken to properly construct and work with the invariants of interest.

The Bloch bundle picture is advantageous as it links much more 
clearly to the geometric explanations of the integer quantum Hall 
effect by~\cite{TKNN} and others, explicitly relating physical 
quantities to homology theories and pairings. The limitation of 
such a viewpoint is that it cannot fully take into account the 
situation with a magnetic field present, which may include 
systems with sublattice/chiral symmetry. 
In such a picture, one would need to perform an analysis 
similar to that of Bellissard for the integer quantum Hall effect 
and handle the noncommutative Brillouin zone. It is also 
generally acknowledged that to work disorder into the Bloch 
bundle viewpoint requires a further development of the 
noncommutative method of Bellissard et al~\cite{Bellissard94}.

\subsubsection*{Chern numbers, spin-Chern numbers and disorder (Prodan, Schulz-Baldes)}
A concerted attempt to adapt the ideas and constructions of 
Bellissard's noncommutative Brillouin zone and Chern 
numbers into the general insulator picture has been made by 
Prodan and Schulz-Baldes in several 
papers~\cite{Prodan10, Prodan11,Prodan14, SchulzBaldes13, SchulzBaldes13b}. 
Part of this process involves showing how Bellissard's cocycle 
formula for the Hall conductance has natural generalisations to 
higher dimensions~\cite{PLB13, PSB14}. 

Another important aspect of Prodan and Schulz-Baldes' 
work has been defining the so-called spin-Chern numbers. 
Roughly speaking, certain $2$-dimensional systems with 
odd time-reversal symmetry can be split into the 
$\pm 1$ eigenspaces of a Pauli matrix representing a 
component of spin (say $s^z$). 
One can then take the Fermi projection $P$, 
restrict it to the $+1$ or $-1$ eigenspace 
of the spin, $P_\pm$, and take the Chern number 
of this restriction $\mathrm{Ch}(P_\pm)$, 
denoted the spin-Chern number. Due to the 
time-reversal symmetry, $\mathrm{Ch}(P)=0$, but 
the two separate spin-Chern numbers may be non-zero. Hence one can 
interpret these invariants as capturing the conductance of the 
spin-up and spin-down currents of the quantum spin-Hall effect. 
One can also associate a $\Z_2$ number to systems with 
spin-Chern numbers by considering $\mathrm{Ch}(P_\pm)\,\mathrm{mod}\,2$. 
It is a key result of~\cite{SchulzBaldes13b} that the 
$\Z_2$ number associated to systems with well-defined spin-Chern numbers
can be related to real $K$-theory, specifically $KO_2(\R)$, 
which is the group that `classifies' $2$-dimensional strong 
topological phases with odd time reversal symmetry.

The use of noncommutative methods also means that the 
models considered by Schulz-Baldes and Prodan are among 
the few that allow disorder to be included in the system 
(see~\cite{SchulzBaldes13} for more details).

Prodan and Schulz-Baldes are able to import techniques 
from complex pairings and cyclic cohomology to define 
their $\Z_2$ invariant. However, we emphasise that
torsion pairings arising from real $K$-theory can \emph{not} 
in general be described as a complex pairing modulo 2. 
See for example the discussion in Witten~\cite[Section 3.2]{Witten15}.
Let us briefly explain why these difficulties arise.

Early results of Connes show that the pairing of 
$K$-theory classes with cyclic cocycles
is the same as the index pairing of $K$-theory classes 
with finitely summable Fredholm modules representing 
$K$-homology classes over complex algebras~\cite{Connes85}. 
However, such a 
relation breaks down in the case of torsion invariants, which are 
common in the $K$-groups of real/Real algebras. An 
extra argument is therefore required to link torsion invariants 
from real $K$-theory with pairings involving cyclic cohomolgy.
We do however note that the complex Chern numbers of systems of 
arbitrary dimension considered by~\cite{PLB13, PSB14} 
and in \cite{Bourne thesis} can be applied to insulator 
systems where only sublattice/chiral symmetry is considered.

Since in general we must avoid using the pairings in cyclic cohomology and homology, we 
deal directly with the $K$-theory and $K$-homology groups directly for complex, 
and Real/real algebras, since these pairings can detect torsion invariants. 
This is the picture adopted in the later 
works~\cite{DNSB14b, GSB15}. We also adopt this viewpoint, 
but from the perspective of $KK$-theory (which is necessary to consider the bulk-edge problem).

We briefly mention the recent paper of Katsura and Koma~\cite{KK15}, 
which also uses (complex) noncommutative methods and obtains similar 
results to Schulz-Baldes and Prodan, particularly~\cite{SchulzBaldes13b}. 
Katsura and Koma associate a $\Z_2$ invariant to $2$-dimensional systems 
with odd time-reversal symmetry and show it is stable under perturbations 
and disorder, the Kane-Mele invariant being an important example. Implicit 
in~\cite{KK15} is that the defined $\Z_2$ invariant is a particular 
representation of the Clifford index that classifies $KO_2(\R)\cong \Z_2$ 
(see~\cite[Chapter \Rmnum{2}.7]{SpinGeometry} for more on the Clifford index). 
We would like to extend such a picture to more general models and symmetries.

\subsubsection*{Symmetry groups and equivariant $K$-theory (Freed-Moore, Thiang)}
So far our various symmetries have been considered on a case by case basis 
with no unifying theory linking systems together as Kitaev outlined. 
Such a theory in the commutative setting was developed by 
Freed and Moore~\cite{FM13}, and then generalised to possibly 
noncommutative algebras by Thiang~\cite{Thiang14}.

The paper by Freed and Moore is very long and detailed so we will 
only give the most basic of summaries. The symmetries of 
interest to us (time-reversal, charge-conjugation and sublattice) 
are put together in a symmetry group $G$. Then, symmetry 
compatible Hamiltonians correspond to projective unitary/anti-unitary 
representations of $G$ (or a subgroup thereof). Using the 
Bloch-bundle viewpoint to derive topological invariants of the 
system under consideration, the quantities of interest can be 
derived by looking at the equivariant $K$-theory of subgroups of $G$. 
In certain cases, lattice symmetries and the crystallographic group of 
the lattice of the sample can also be incorporated, giving rise to 
possibly twisted equivariant $K$-theory classes and invariants.

The work of Thiang showed how Freed-Moore's constructions can be 
carried out in the noncommutative setting. In particular, Thiang links 
symmetry data to Clifford algebras and constructs a homology theory 
similar to Karoubi's $K^{p,q}$-theory (see~\cite[Chapter \Rmnum{3}]{Karoubi}) 
that encodes these symmetries. Such a construction means that the 
classification by Kitaev and Ryu et al. (also called the $10$-fold way) can be described in a unified framework.

Freed-Moore and Thiang's work allows all the symmetry data to be 
considered on an equal footing and gives a 
mathematical argument for Kitaev's classification. The work of Thiang 
in particular opens the door to further research as it provides a concrete 
framework to consider disordered systems and impurities. The main 
limitation is that the theory deals solely with a bulk system and $K$-theory. 
The use of $K$-homology or a system with edge is not considered.

Recent work by Kellendonk also provides a systematic study of 
symmetry compatible Hamiltonians and algebras with van Daele's formulation of 
$K$-theory~\cite{Kellendonk15}. Using slightly different grading structures to 
Freed-Moore and Thiang, Kellendonk derives Kitaev's $K$-theoretic classification for 
discrete or continuous systems with disorder, and a Hamiltonian with spectral gap. Again, 
$K$-homology or systems with boundaries are not considered.

\subsubsection*{$KR$-Theory and pairings (Grossmann-Schulz-Baldes)}
A recurring characteristic of the literature on topological insulators, both physical 
and mathematical, is that the links to topology are solely discussed via $K$-theory. 
However, as shown in Equation \eqref{eq:Bellissard_Hall_conductance}, 
the expression for the Hall conductance is not just a $K$-theory construction, 
but a pairing (i.e. Kasparov product) between a $K$-theory class and a 
$K$-homology class coming from a particular spectral triple or Fredholm module. 
Most literature on topological insulators does not consider this extra 
$K$-homological information, though an exception are the papers 
of Schulz-Baldes and co-authors~\cite{DNSB14b, GSB15}.

De Nittis, Grossmann and Schulz-Baldes show that a discrete 
condensed matter system with additional symmetries naturally 
gives rise to a Real spectral triple in the sense of 
Connes~\cite{Connes95, ElementsNCG} and represents a 
$KR$-homology class. De Nittis and Schulz-Baldes consider the 
$2$-dimensional case~\cite{DNSB14b} and Grossmann-Schulz-Baldes 
generalise this to arbitrary dimension~\cite{GSB15}. 
In particular,~\cite{DNSB14b,GSB15} show that the 
Fermi projection of a symmetry compatible Hamiltonian 
pairs with the Real spectral triple via an index and it is this 
pairing that gives the various (strong) classification groups of Kitaev, Freed-Moore and Thiang.

De-Nittis, Grossmann and Schulz-Baldes's work provides a 
useful picture of the bulk-theory of insulators. Working the 
bulk-edge correspondence into such a framework 
remains to be done. Equally, the work of Thiang and 
Grossmann-Schulz-Baldes may be related under the 
broader framework of $KK$-theory although we will not do so here.

\subsection{The bulk-edge correspondence}
So far our discussion has been focused on how single-particle Hamiltonians with certain symmetries give rise to topological invariants, but the way in which these properties are physically realised is a key aspect of insulator materials. Namely, the observables that are measured in experiment are said to be carried on the edge or boundary of a sample. So on the one hand, we have a Hamiltonian acting on the whole space, often assumed to be translation invariant, which gives topological properties of the material via the Bloch bundles over the Brillouin zone (or a noncommutative analogue of this). On the other hand, there is also a current or similar observable concentrated at the edge of the sample that is said to be `topologically protected' by the internal symmetries of the Hamiltonian. Loosely speaking, the relationship between the topological properties of the bulk Hamiltonian and edge behaviour is the bulk-edge correspondence of topological insulator materials.

While a mathematical understanding of the bulk-edge correspondence for systems with anti-linear symmetries is still in development, there have been a few important contributions. Firstly there was the work of~\cite{ASV13, SchulzBaldes13}, who consider $2$-dimensional time-reversal invariant systems and prove a bulk-edge correspondence using the spin-Chern perspective and an argument using transfer matrices. A $2$-dimensional bulk edge correspondence for systems with time-reversal symmetry is also considered in~\cite{GP13}. By using more elementary functional analytic techniques, Graf and Porta reproduce the result of~\cite{ASV13, SchulzBaldes13} for a broader class of possible Hamiltonians.

These are both useful results and important contributions to the literature, though the link between the bulk-edge picture described in these papers and the $K$-theoretic classification is very difficult to establish, though the two should be compatible. 

Section 7 of~\cite{Loring15} considers the bulk-edge correspondence in arbitrary dimension. What is difficult to determine is the link between Loring's argument and the work of Grossmann, Schulz-Baldes and Thiang as well as the bulk-edge picture developed by Kellendonk et al.

Papers by Mathai and Thiang establish a $K$-theoretic 
bulk-edge correspondence for complex systems and real systems with time-reversal 
symmetry~\cite{MT15b, MT15c}. These 
papers use a short-exact sequence to link bulk and edge systems as considered 
by~\cite{KR06,SBKR02,KSB04b} in the case of the quantum Hall effect. 
One can then check that the invariants of interest (including 
torsion invariants for time-reversal symmetric systems) pass from bulk to edge in the 
Pimsner-Voiculescu sequence in complex or real $K$-theory. Mathai and Thiang also use 
real and complex T-duality to show in a variety of examples
that when the boundary map in $K$-theory is T-dualised, 
the map can be expressed as a conceptually simpler restriction map. 
In the real case, the Kane-Mele $\Z_2$ invariant is also identified with the 
$2$nd Stiefel-Whitney class under T-duality. 
Our work in progress directly computes the 
$K$-theoretic boundary map for systems with arbitrary symmetry type by taking the 
Kasparov product with (unbounded) representatives of the extension class in 
$KK$-theory~\cite{BKR}.

After the submission of this work, further developments have appeared on 
the mathematical aspects of the bulk-edge correspondence. In particular we 
briefly highlight the work of Kubota~\cite{Kubota15b} and the book by 
Prodan and Schulz-Baldes~\cite{PSBbook}.

The book by Prodan and Schulz-Baldes is a complete analysis of 
complex disordered topological insulator systems. This includes bulk invariants, 
edge invariants, their relation to $K$-theory as well as their 
equality for systems with boundary~\cite{PSBbook}. The book also shows how the 
edge pairings can be linked to conductivity tensors and the 
physical system under consideration. Real and torsion invariants 
are not considered.

Kubota provides a noncommutative extension of Freed--Moore's work.
Hamiltonians 
compatible with a twisted symmetry group are associated to a 
class in $\Z_2$-graded twisted equivariant $K$-theory, which is 
defined using an extension of $KK$-theory~\cite{Kubota15a}. The 
case of invariants related to the $CT$-symmetry group are a simple 
example. Kubota uses the uniform Roe algebras and relates bulk and 
edge systems by the boundary map in the coarse Mayer--Vietoris 
exact sequence in $K$-theory~\cite{Kubota15b}. The $K$-theory of 
uniform Roe algebras is in general very hard to compute though 
can be simplified by the coarse Baum--Connes map. The use of 
coarse geometry allows for systems with impurities or rough 
edges to be considered.

\section{Bulk theory} \label{sec:bulk_theory}

\subsection{Symmetry types and representations}
In our basic setup, we consider a self-adjoint single-particle Hamiltonian $H$ acting on a complex Hilbert space $\calH$. We work in the tight-binding approximation so that $\calH$ will usually take the form $\ell^2(\Z^d)\otimes\C^N$, where $d$ captures the dimension we are considering and $N$ encodes any internal degrees of freedom coming from properties such as spin or the structure of our lattice. Our outline of the basic symmetries is quite similar to that discussed in, amongst others,~\cite{DNSB14b,GSB15}.

We consider the symmetries of our Hamiltonian $H$. The symmetries of interest to us are time-reversal symmetry, charge-conjugation  symmetry (also called particle-hole symmetry) and sublattice symmetry (also called chiral symmetry). Each of these symmetries is  given by an  involution
and we use the notation, $T\equiv$  time-reversal, $C\equiv$ charge-conjugation and $S\equiv$ sublattice. These are not independent
but  generate the  $CT$-symmetry group $\{1,T,C,CT\} \cong \Z_2\times\Z_2$, where $S=CT=TC$.

\begin{defn} \label{def:symmetry_compatible_hamiltonian}
A Hamiltonian $H$ acting on a complex Hilbert space $\calH$ respects time-reversal and/or charge-conjugation and/or sublattice symmetry if there are complex anti-linear operators $R_{T}$ and/or $R_{C}$ and/or a complex-linear operator $R_{S}$ acting on $\calH$ such that $R_T^2,R_C^2,R_S^2\in\{\pm 1_\calH\}$ and
\begin{align} \label{eq:symmetry_defining_properties_for_Hamiltonian}
   &R_T H R_T^* = H,   &&R_C H R_C^* = -H,  &&R_S H R_S^* = -H
\end{align}
In the case of $R_T$ and $R_C$, our Hamiltonian is said to have even (resp. odd) symmetry if $R^2 = 1$ (resp. $R^2 = -1$). 
\end{defn}
Because $R_S$ is complex-unitary, the sign of its square is irrelevant (in the same way that the Clifford algebra on one generator $\C\ell_1$ may have a generator that squares to $+1$ or $-1$). We note that a Hamiltonian may only respect a single symmetry. However, if $H$ is compatible with   two symmetries, then by the underlying group structure it is compatible with the third symmetry. We will more explicitly examine the link between symmetry compatible Hamiltonians and group representations in Section \ref{subsec:Linking_symmetries_to_Clifford_algebras}.

There is no general form that the symmetry operators $R_T$, $R_C$ and $R_S$ are forced to take apart from the properties outlined in Definition \ref{def:symmetry_compatible_hamiltonian}. Instead they need to be determined by the properties of the example under consideration. However  the conjugate-linear operators $R_T$ and $R_C$, as operators acting on a complex Hilbert space, anti-commute with the Real involution given by complex conjugation. 

\begin{example}[Anti-linear symmetries via complex conjugation] \label{ex:symmetry_repn_via_complex_conjn}
We consider the Hilbert space $\ell^2(\Z^d)\otimes\C^{2N}$ and define the operator
$$  R = \begin{pmatrix} 0 & \calC \\ \eta\calC & 0 \end{pmatrix}, $$
where $\calC$ is complex conjugation and $\eta\in\{\pm 1\}$. At this stage we are not specifying whether $R$ represents the time-reversal or charge-conjugation involution. We note that $R^2=\eta 1_{2N}$ so $R$ can represent an even or odd symmetry depending on the sign of $\eta$. Given an operator $a\in\calB[\ell^2(\Z^d)\otimes \C^N]$ we define the operator $\ol{a}=\calC a\calC$. One computes that
\begin{equation} \label{eq:transform_2x2_matrix_by_complex_conjugation_involution}
  R\begin{pmatrix} a & b \\ c & d \end{pmatrix}R^* = \begin{pmatrix} \ol{d} & \eta\ol{c} \\ \eta\ol{b} & \ol{a} \end{pmatrix}.
\end{equation}

Consider the case that $R$ is implementing a time-reversal involution. By Equation \eqref{eq:transform_2x2_matrix_by_complex_conjugation_involution}, an operator acting on $\ell^2(\Z^d)\otimes \C^{2N}$ will be time-reversal symmetric if it takes the form $\begin{pmatrix} a & b \\ \eta\ol{b} & \ol{a} \end{pmatrix}$. Such an operator will also have to be  self-adjoint if it is to be a time reversal symmetric Hamiltonian.

Next consider the case that $R$ is representing the charge-conjugation involution. An operator $A$ is symmetric under charge-conjugation if $RAR^* = -A$, so Equation \eqref{eq:transform_2x2_matrix_by_complex_conjugation_involution} tells us that $A$ must be of the form $\begin{pmatrix} {a} & {b} \\ -\eta\ol{b} & -\ol{a} \end{pmatrix}$.

Recall the Dirac-type operator from the spectral triple of the integer quantum Hall effect (Proposition \ref{prop:qH_spec_trip}), $X = \begin{pmatrix} 0 & X_1-iX_2 \\ X_1+iX_2 & 0 \end{pmatrix}$, where
$X_1,X_2$ act on $\ell^2(\Z^2)$ as position operators. We see that $X$ is even time-reversal symmetric (that is $RX R^* = X$ with $\eta=1$) or has odd charge-conjugation symmetry ($RXR^* = -X$ with $\eta = -1$) depending on the symmetry that the  involution $R$ is representing.
\end{example}

\begin{example}[Symmetries via spatial involution] \label{ex:odd_symmetry_with_spatial_involution}
We start with the space $\ell^2(\Z^d)$ and define the anti-linear operator $J$ such that $(J\lambda)(x) = \ol{\lambda(-x)}$ for $\lambda\in\ell^2(\Z^2)$. We define on $\calH = \ell^2(\Z^d)\otimes\C^{2N}$ the operator
$$ R = \begin{pmatrix} 0 &  J \\ \eta J & 0 \end{pmatrix}  $$
with $R^2=\eta 1_{2N}$ as before. As a transformation on operators acting on $\ell^2(\Z^d)\otimes\C^{2N}$, one computes that
$$ R\begin{pmatrix} a & b \\ c & d \end{pmatrix} R^* = \begin{pmatrix} JdJ & \eta JcJ \\ \eta JbJ & JaJ \end{pmatrix}. $$

We again consider the case that $R$ models the time-reversal or 
particle-hole involution. Operators that are time-reversal invariant 
under conjugation by $R=R_T$ have the general form 
$\begin{pmatrix} a & b \\ \eta JbJ &  JaJ \end{pmatrix}$, 
whereas charge-conjugation symmetric operators under 
$R=R_C$ are of the form $\begin{pmatrix} a & b \\ -\eta JbJ &  -JaJ \end{pmatrix}$.

Considering again the integer quantum Hall spectral triple with 
Dirac-type operator $X$, we first note that $JX_kJ = -X_k$ 
and $J(\pm i X_k)J = \pm i X_k$ for the position operators $X_k$, $k=1,2$. Therefore we have that 
$$  
R \begin{pmatrix} 0 & X_1-iX_2 \\ X_1+iX_2 & 0 \end{pmatrix} R^* 
= \begin{pmatrix} 0 & \eta(-X_1+iX_2) \\ \eta(-X_1-iX_2) & 0 \end{pmatrix}. 
$$
This implies that $X$ now has odd time reversal symmetry and even charge-conjugation symmetry.
\end{example}

\begin{remark}[Time-reversal and charge-conjugation as $0$-dimensional phenomena] \label{remark:PT_symmetries_are_0_dimensional}
The example of the integer quantum Hall Dirac-type 
operator shows that changing how we represent the 
involutions $R_T$ and $R_C$ may change whether 
an operator has a particular symmetry type. This  
indicates that the spatial involution is bringing extra 
data into our system (namely, that we have a 
$d$-dimensional sample with $d>0$). By comparison, 
time-reversal and charge-conjugation involutions can 
exist in $0$-dimensional samples and do not need the 
extra information that spatial involution does. We 
emphasise that systems with anti-linear symmetries 
defined using spatial involution are topologically inequivalent 
to systems with anti-linear symmetries defined from complex 
conjugation (see~\cite{MT15} for more detail on the inequivalence of symmetry types).
\end{remark}

\begin{example}[Chiral symmetry]
In most examples in the literature, the chiral/sublattice symmetry 
involution is represented by the matrix 
$R_S = \begin{pmatrix}1_N & 0 \\ 0 & -1_N \end{pmatrix}$ 
on $\ell^2(\Z^d)\otimes\C^{2N}$, so a self-adjoint Hamiltonian 
$H$ is chiral symmetric if $H= \begin{pmatrix} 0 & h \\ h^* & 0 \end{pmatrix}$.
\end{example}

An important observation is that if $H$ obeys a particular symmetry 
and the Fermi level $\mu$ is in a gap of the spectrum of 
$H$ (we can assume without loss of generality that $\mu=0$), 
then the `spectrally flattened' Hamiltonian $\text{sgn}(H) = H|H|^{-1}$ 
also obeys this symmetry.

\subsection{Symmetries, group actions and Clifford algebras} \label{subsec:Linking_symmetries_to_Clifford_algebras}
We have briefly explained the symmetries that arise in our insulator systems but we would like a more structural understanding of how these symmetries fit into a unifying picture. Here the recent work of Thiang as developed in~\cite{Thiang14,Thiang14b,MT15}, and based on ~\cite{FM13}, is
applicable. One of the key insights in~\cite{FM13, Thiang14} is to see that a symmetry compatible Hamiltonian with a spectral gap $H$ can be expressed as a graded projective unitary/anti-unitary representation of the finite symmetry group $G\subset\{1,T,C,CT\}\cong\Z_2\times\Z_2$.

\begin{defn}
Let $G$ be a finite group with $\theta_g$ a map on $\calH$ for every $g\in G$ and $\varphi:G\to \{\pm 1\}$ a homomorphism. The triple $(G,\varphi,\sigma)$ is a projective unitary/anti-unitary (PUA) representation if $\theta_g$ is unitary (resp. anti-unitary) and $\varphi(g)=1$ (resp. $-1$) and $\theta_{g_1}\theta_{g_2}=\sigma(g_1,g_2)\theta_{g_1g_2}$ with $\sigma:G\times G\to \T$ a  $2$-cocycle  satisfying
$$  \sigma(g_1,g_2)\sigma(g_1g_2,g_3) = \sigma(g_2,g_3)^{g_1} \sigma(g_1,g_2g_3), \quad g_1,g_2,g_3\in G, $$
where for $z\in \T$, $z^{g}=z$ if $\varphi(g)=1$ and $z^g = \ol{z}$ if $\varphi(g)=-1$.
\end{defn}

We can now re-formulate the definition of a symmetry compatible Hamiltonian in terms of group representations.
\begin{defn} \label{def:Hamiltonian_compatible_with_G}
Given a projective unitary/anti-unitary representation $(G,\varphi,\sigma)$ and a gapped self-adjoint Hamiltonian $H$ acting on a complex Hilbert space $\calH$, we say that $H$ is compatible with $G$ if there is a (continuous) homomorphism $c:G\to\{\pm 1\}$ such that
\begin{equation} \label{eq:graded_PUA_repn_and_hamiltonian}
  \theta_g H = c(g)H\theta_g \hspace{0.5cm} \text{for all }g\in G. 
\end{equation}
\end{defn}

\begin{remark}
Under our assumptions $0\notin\sigma(H)$, so we can deform a symmetry-compatible $H$ to its phase $H|H|^{-1}=\text{sgn}(H)$ without changing Equation \eqref{eq:graded_PUA_repn_and_hamiltonian}. This means that $\Gamma = \text{sgn}(H)$ is acting as a grading of our PUA representation. Therefore, we say that a symmetry compatible Hamiltonian on $\calH$ is precisely realised as a graded PUA representation $(G,c,\varphi,\sigma)$ on $\calH$ with grading $\Gamma = \text{sgn}(H)$. The map $\varphi$ determines if the symmetry involution $\theta_g$ is represented unitarily or anti-unitarily and the map $c$ determines if the involution has even or odd grading. We emphasise that the grading of a symmetry involution $\theta_g$ as even or odd is different from whether the symmetry is denoted even or odd, which comes from whether $\theta_g^2=1$ or $-1$ respectively.
\end{remark}

Our definition of a symmetry compatible 
Hamiltonian may apply to any finite group $G$, 
though we are interested in a subgroup $G$ of the 
$CT$-symmetry group $\{1,C,T,CT\} \cong \Z_2\times\Z_2$. 
Equation \eqref{eq:symmetry_defining_properties_for_Hamiltonian} 
and the surrounding discussion tells us that a symmetry 
compatible Hamiltonian $H$  can be expressed as a 
PUA representation of a subgroup $G$ of the symmetry
group $\{1,C,T,S=CT\}$ on the 
Hilbert space $\calH=\ell^2(\Z^d)\otimes\C^N$ 
with $\theta_g=R_g$, $\Gamma = \mathrm{sgn}(H)$ and
\begin{align*}
   &(\varphi,c)(T) = (-1,1),  &&(\varphi,c)(C)=(-1,-1),   &&(\varphi,c)(S) = (1,-1).
\end{align*}

Representations of $G$ are in $1$-$1$ correspondence with 
representations of the real or complex group $C^*$-algebra 
$C^*(G)$. From the perspective of Kasparov theory we would 
like to link the algebra $C^*(G)$ with real or complex 
Clifford algebras as such algebras play a fundamental role in $KK$-theory.

\begin{prop}[\cite{FM13}, Appendix B; \cite{Thiang14}, Section 6] 
\label{prop:PUA_rep_gives_Clifford_rep}
Let $G$ be a subgroup of the symmetry group $\{1,T,C,S=CT\}$ 
with $G\neq\{1,S\}$. If a Hamiltonian $H$ acting on 
$\calH$ is compatible with $G$, then there is a graded 
representation of a real Clifford algebra on $\calH$ (or $\calH\oplus\calH$). 
If $G=\{1,S\}$, then there is a graded complex Clifford representation. 
The representations are summarised in 
Table \ref{table:PT_classes_repn_table} up to stable isomorphism.
\end{prop}

The natural grading of real Clifford algebras gives  
all generators of the Clifford algebra odd degree. Therefore all generators 
of a Clifford representation must be odd with respect to the grading $\Gamma = \mathrm{sgn}(H)$. 

\begin{proof}
The proof proceeds on a case by case basis. We first use~\cite[Proposition 6.2]{Thiang14} 
to `normalise' the twist $\sigma$ of the PUA representation so that 
the operators $R_C$ and $R_T$ commute and 
$R_C R_T = R_{CT}$.
For the full symmetry group $G=\{1,C,T,CT\}$,  
we use the operators $R_g$ for $g\in G$ 
and consider the real algebra generated by 
$\{R_C, iR_C, iR_{C}R_{T}\}$. One checks that 
these generators have odd grading under $\Gamma$, 
mutually anti-commute and are self-adjoint (resp. skew-adoint) 
if they square to $+1$ (resp. $-1$). Therefore the real algebra 
generated by $\{R_C, iR_C, iR_{CT}\}$ is precisely a 
graded representation of a particular real Clifford algebra 
$C\ell_{r,s}$ with grading $\Gamma = \mathrm{sgn}(H)$ 
(our notation for Clifford algebras is explained in Appendix \ref{sec:RealKasTheory}).

Next, we consider the subgroup $\{1,C\}$, to which we assign 
the real algebra generated by $\{R_C, iR_C\}$ and graded by $\mathrm{sgn}(H)$. 

Representations of the subgroup $\{1,S\}$ give rise to a 
representation generated by $R_S$ with grading $\mathrm{sgn}(H)$. 
Because $R_S$ acts complex-linearly, we may consider the 
complex span of $R_S$ as acting on $\calH$. Hence the 
representation generated by $R_S$ is a graded representation of $\C\ell_1$.

The case of the subgroup $\{1,T\}$ is a little different as 
$R_T$ commutes with $\mathrm{sgn}(H)$. 
For the case that $R_T^2=1$, $R_T$ defines a 
Real structure on the Hilbert space and gives no additional 
Clifford generators. If $R_T^2=-1$, then $R_T$ 
defines a quaternionic structure on $\calH$ under the 
identification $\{i,j,k\}\sim \{i,R_T,iR_T\}$. 
There is an equivalence between a graded 
quaternionic vector space and a graded action of $C\ell_{4,0}$ on $\calH\oplus \calH$. 
Specifically, we take $\calH\oplus\calH$ and the real span of the Clifford generators
\begin{align*}
   \left\{\begin{pmatrix} 0 & 1 \\ 1 & 0 \end{pmatrix}, 
   \begin{pmatrix} 0 & -i \\ i & 0 \end{pmatrix}, 
   \begin{pmatrix} 0 & -R_T \\ R_T & 0 \end{pmatrix}, 
   \begin{pmatrix} 0 & -iR_T \\ iR_T & 0 \end{pmatrix}\right\}, 
   \quad \Gamma = \begin{pmatrix} \mathrm{sgn}(H) & 0 \\ 0 & -\mathrm{sgn}(H) \end{pmatrix}.
\end{align*}
Therefore, the subgroup $\{1,T\}$ gives rise to a 
graded representation of $C\ell_{0,0}$ or $C\ell_{4,0}$.
\end{proof}

\begin{table} 
    \centering
    \begin{tabular}{ p{2cm} | c c |  p{3.8cm} }
       \multirow{3}{2cm}{Symmetry generators} & \multirow{3}{*}{$R_C^2$} & \multirow{3}{*}{$R_T^2$}  & \multirow{3}{4cm}{Graded Clifford representation (up to stable isomorphism)} \\ 
        & &  &  \\ & &  &  \\  \hline
       $T$ &  & $+1$ &   $C\ell_{0,0}$ \\
       $C,T$ & $+1$ & $+1$ & $C\ell_{1,0}$ \\
       $C$ & $+1$ &  &  $C\ell_{2,0}$ \\
       $C,T$ & $+1$ & $-1$  & $C\ell_{3,0}$ \\
       $T$ &  & $-1$  & $C\ell_{4,0}$ \\
       $C,T$ & $-1$ & $-1$  & $C\ell_{5,0}$ \\
       $C$ & $-1$ &   & $C\ell_{6,0}$ \\
       $C,T$ & $-1$ & $+1$ & $C\ell_{7,0}$ \\ \hline \hline 
       N/A &  &  &  $\C\ell_{0}$ \\
       $S$ & \multicolumn{2}{c|}{$R_S^2=1$}  &  $\C\ell_1$
    \end{tabular}
    \caption{Symmetry types and their corresponding graded Clifford representations~\cite[Table 1]{Thiang14}. \label{table:PT_classes_repn_table}}
\end{table}
\begin{cor} \label{lem:group_algebra_of_G_is_ME_to_Clifford}
Let $G$ be a subgroup of $\{1,T,C,CT\}$. The real or complex twisted 
group $C^*$-algebra $C^*_\sigma(G)$ 
is stably isomorphic to $C\ell_{n,0}$ or $\C\ell_n$, where 
$\sigma$ is the twist from the PUA representation and 
$n$ is determined by Table 
\ref{table:PT_classes_repn_table}.
\end{cor}
\begin{proof}
Follows from Proposition \ref{prop:PUA_rep_gives_Clifford_rep}.
\end{proof}

\begin{remark}[The $10$-fold way]
A graded PUA representation of $\{1,T,C,CT\}$ 
gives rise to the real Clifford generators $\{R_C, iR_C, iR_{CT}\}$. 
These generators represent four different Clifford algebras 
determined by the sign of $R_C^2$ and $R_T^2$. 
Similarly, the representations of the subgroup $\{1,C\}$ 
give representations for two real Clifford algebras 
generated by $\{R_C, iR_C\}$ and vary depending on 
whether $R_C^2 = \pm 1$. Graded representations of 
$\{1,S\}$ correspond to the Clifford algebra 
$\text{span}_\C\{R_S\}\cong \C\ell_1$, which is the 
same whether $R_S^2=\pm 1$ (again, these 
representations come with the grading $\Gamma=\mathrm{sgn}(H)$). 
A Hamiltonian compatible with the symmetry group $\{1,T\}$ 
gives rise to two real Clifford algebras depending on 
whether $R_T^2 = \pm 1$. In total, we have nine 
possible representations of symmetry subgroups 
as \emph{distinct} Clifford algebras and a lack of any 
symmetry gives us one more possibility. This is the 
well-known `$10$-fold way' that arises when we 
consider symmetries of this kind (see for example~\cite{SRFL08}).

Because we are interested in the link between 
Clifford representations and $KK$-theory, we may 
choose representations up to stable isomorphism, 
where $C\ell_{r+1,s+1}\cong C\ell_{r,s}\hat\otimes M_2(\R)$ 
for real Clifford algebras and $\C\ell_{n+2}\cong \C\ell_{n}\hat\otimes M_2(\C)$ 
for complex algebras. We summarise the results in Table \ref{table:PT_classes_repn_table}.

We note that in Table \ref{table:PT_classes_repn_table}, 
each symmetry type gives rise to a distinct graded Clifford representation. 
Therefore (up to stable isomorphism), the process is reversible. 
That is, given a graded representation of $C\ell_{n,0}$ or $\C\ell_n$, 
we may think of this representation as encoding the 
symmetries of a subgroup of the $CT$-group that are 
compatible with a gapped Hamiltonian with $\Gamma = \mathrm{sgn}(H)$.
\end{remark}

\subsection{Internal symmetries and $KK$-classes} \label{subsec:symmetries_and_KKO}
In the previous section, we outlined how symmetry-compatible gapped self-adjoint Hamiltonians give rise to a graded $\ast$-representation of $C\ell_{n,0}$ or $\C\ell_n$ with the number $n$ determined (up to stable isomorphism) by the symmetries present and whether they are even or odd. Our next task is to relate this characterisation to the $K$-theory of our observable algebra.

Before we specify our observable algebra, we must first specify the class of of bulk Hamiltonians our method can be adapted to. As observed in the
integer quantum Hall example (cf.~\cite{Bellissard94, CM96}), in order to study the geometry and topology of the Brillouin zone, we require an algebra of observables larger than the algebra generated by the Hamiltonian (or its resolvent).

\begin{assumption} \label{def:bulk_Hamiltonian_condition}
Unless otherwise stated, we will assume the Hamiltonians we consider act on $\ell^2(\Z^{d})\otimes\mathbb{C}^N$ 
and are represented by matrices whose entries are either finite polynomials of (possibly twisted) shift operators or infinite polynomials with Schwartz-class coefficients. The Hamiltonians also have a spectral gap containing the Fermi level.

If $H$ is compatible with the symmetry group $G$, a subgroup of $\{1,T,C,CT\}$, then we also assume that the symmetry action $H\mapsto R_gHR_g^*$ extends to an action on the algebra generated by the (twisted) shift operators that generate $H$.
\end{assumption}

We note that essentially all tight-binding (discrete) model Hamiltonians without disorder satisfy our criterion. We consider the algebra
generated by the shift operators that give rise to $H$ and act on $\ell^2(\Z^d)\otimes \mathbb{C}^N$. 
We require that the action of symmetries on the Hamiltonian $H$ extends 
to the observable algebra (Assumption \ref{def:bulk_Hamiltonian_condition}) 
in order to determine symmetry properties of the whole Brillouin zone. 
Such an assumption is required in the case of abstract representations 
of the symmetry group $G\subset \{1,T,C,CT\}$, though is easily 
satisfied in the common representations that arise in examples 
(e.g. symmetry involutions defined by complex conjugation or spatial involution).

Because we are starting with operators on a complex Hilbert space, 
if there are anti-linear symmetries, our algebra of interest is a \emph{real} subalgebra of a
\emph{complex} algebra. Specifically, we take the complex $C^*$-algebra generated by the shift operators and then the real subalgebra that is
invariant under the involution $a^\tau = \calC a \calC$ with $\calC$ complex conjugation on $\calH$. Such a condition requires the 
shift operators to be untwisted by an external magnetic field. 
This is unsurprising in the case of a time-reversal
symmetric Hamiltonian, which cannot have an external 
magnetic field. Systems with charge-conjugation symmetry that contain an external
magnetic field require a more careful treatment 
of the algebras of interest. In the interest of brevity, we will avoid this issue and
instead assume that if there are anti-linear symmetries present, the magnetic flux vanishes. 
See~\cite{Kellendonk15} for more detail on such problems.

In the case that $G=\{1,T,C,CT\}$, $\{1,T\}$ or $\{1,C\}$, 
we take our bulk algebra to be the subalgebra of the matrix algebra of 
shift operators that generate $H$. Namely, we denote $A\subset M_N(C^*(\Z^d))$ with
$C^*(\Z^d) \cong C^*(S_1,\ldots,S_d)$ the
real $C^*$-algebra generated by shift operators. We note that $A$ acts on both the complex Hilbert space
$\ell^2(\Z^d)\otimes\C^N$ and the real Hilbert space 
$\ell^2(\Z^d)\otimes\R^M$, which will be important when we construct real spectral
triples. If $G\subset\{1,S\}$ our shift operators may be 
twisted and we denote the complex algebra $A_\C \subset M_N(A_\phi^d)$ with $A_\phi^d$ the $d$-dimensional
rotation algebra (of course we may also take $\phi=0$).

We also note that our algebras are modelling systems without 
disorder. We have started with this basic model for simplicity and to make our
constructions as clear as possible. We will comment on some 
extensions to the case of algebras modelling weak disorder in Section
\ref{subsec:remarks_on_disorder}, though delay a full investigation to future work.

\subsubsection{The symmetry class}

We first construct a Kasparov module, real or complex, which classifies symmetry 
compatible Hamiltonians by the associated $C^*$-algebraic $K$-theory. We employ the 
framework of $KK$-classes as this allows generalisations to other symmetry groups.

Using the action of $G$ on $A\subset M_N(C^*(\Z^d))$ we can 
take the crossed product $A\rtimes G$ (similarly $A_\C \rtimes G$ 
for complex algebras and $G=\{1,S\}$). 
This is one of the 
key reasons we require $A$ to be a real algebra. 
In the case $g=C$ or $T$, the automorphism 
$\alpha_g(a) = R_g aR_g^*$ is complex anti-linear 
and so one can not take the crossed product of this 
automorphism if $A$ is a complex algebra. We can 
realise this crossed product concretely as
$$  
A\rtimes G \cong 
\ol{\text{span}}_\R\!\left\{ \sum_{g\in G} a_g R_g\,:\, a_g\in A\right\}\subset \End_\R(\calH). 
$$
There is a conditional expectation on the crossed-product, 
$\Phi:A\rtimes G\to A$ that has the form 
$$  
\Phi\left( \sum_{g\in G} a_g R_g \right) = a_e \in A. 
$$
The next ingredient that we need to obtain a representative of a
$KKO$-class is a bimodule, specifically a right $A$-module that is also a left 
$C^*(G)$-module. This module has to be equipped with an $A$-valued inner
product (what this means is explained in Appendix \ref{sec:RealKasTheory}), and the next
result uses the expectation $\Phi$ to construct the required bimodule.

\begin{prop}
Let $G$ be a subgroup of the symmetry group $\{1,T,C,CT\}$ 
with $G$ non-trivial and $G\neq\{1,S\}$. Then there is a real 
Hilbert $A$-module $E_A$ defined as the completion of 
$A\rtimes G$ under the norm derived from the inner product 
$( e_1|e_2)_A = \Phi(e_1^* e_2)$ and with right-action given by right-multiplication.
If $G\subset\{1,S\}$ then there is a complex Hilbert $A_\C$-module given by 
(the completion of) $A_\C\rtimes G$ via the same inner-product and norm.
\end{prop}
\begin{proof}
That $\Phi:A\rtimes G \to A$ gives an $A$-valued inner product 
is a check of the definition using the positivity, faithfulness and $A$-bilinearity of $\Phi$.
 To be a Hilbert $A$ module requires $E_A$ to have a right-multiplication 
by $A$ that is compatible with the inner product.  What this means is that for $c \in A$, $( e_1|e_2c)_A=( e_1|e_2)_Ac$.
Let us now see that in the situation described in the statement of the proposition this property holds,
\begin{align*}
  \left(\left. \sum_{g\in G} a_g R_g\right\vert \sum_{h\in G} b_h R_h c \right)_A 
  &= \Phi\!\left(\sum_{g,h\in G} R_g^* a_g^* b_hR_h c\right) \\
    &=  \Phi\!\left(\sum_{g,h\in G} R_g^* a_g^* b_h \alpha_h(c) R_h \right) \\
    &= \sum_{g,h\in G} \delta_{g,h}\alpha_g^{-1}(a^*_g b_h \alpha_h(c))R_g^* R_h
\end{align*}
as $\Phi$ evaluates at the identity. We then simplify
\begin{align*}
  \left(\left. \sum_{g\in G} a_g R_g\right\vert \sum_{h\in G} b_h R_h c \right)_A  
  &= \sum_{g\in G} \alpha_g^{-1}(a_g^* b_g) c = \left(\left. \sum_{g\in G} a_g R_g\right\vert \sum_{h\in G} b_h R_h  \right)_A c.
\end{align*}
We can complete $A\rtimes G$ in the norm 
$\Vert e_1\Vert=\Phi(e_1^* e_1)^{1/2}$ defined from this inner product 
to obtain the real module $E_A$. In fact, since $G$ is finite, $E_A\cong A\rtimes G$ as a linear space. When $G\subset\{1,S\}$, the same proof 
yields a complex module.
\end{proof}

We note that  elements in the crossed product $A\rtimes G$ act on the left on 
$E_A$ by what are termed  adjointable  endomorphisms.
What this means is that the  relation
\begin{equation} \label{eq:multiplication_is_adjointable_on_crossed_product_module}
 ( e_1 e_2| e_3)_A = \Phi(e_2^* e_1^* e_3) = (e_2| e_1^*e_3)_A 
\end{equation}
must hold for for any $e_j\in A\rtimes G, j=1,2,3$. However in this 
instance the previous identity is immediate.
In particular, this means that a left-action by multiplication 
by the real $C^*$-algebra $C^*(G)\subset A\rtimes G$ is adjointable. An analogous observation 
holds for complex algebras and modules if $G\subset \{1,S\}$. 
In the spirit of Proposition \ref{prop:PUA_rep_gives_Clifford_rep}, we obtain the following.

\begin{prop} \label{prop:Clifford_rep_on_module}
Let $H$ be a Hamiltonian satisfying Assumption \ref{def:bulk_Hamiltonian_condition} 
that is compatible with a subgroup $G$ of the symmetry group $\{1,T,C,CT\}$. 
If $G$ is non-trivial and $G\neq\{1,S\}$, then there is 
a graded adjointable representation of $C\ell_{n,0}$ on the 
$C^*$-module $E_A^{\oplus N}$ with grading determined by 
$\mathrm{sgn}(H)$ and $N\in\{2,4\}$. If $G\subset\{1,S\}$, then there is a 
graded representation of $\C\ell_n$, with $n$ determined by 
 Table \ref{table:PT_classes_repn_table}.
\end{prop}
\begin{proof}
We first note that left-multiplication by $R_g$ is adjointable 
for any $g\in G$ by 
Equation \eqref{eq:multiplication_is_adjointable_on_crossed_product_module}. 
The same argument applies to show that the grading $\mathrm{sgn}(H)\in A$ 
is an adjointable operator.

From this point our proof is quite similar to the proof of 
Proposition \ref{prop:PUA_rep_gives_Clifford_rep} and 
is done on a case by case basis. We can once again use~\cite[Proposition 6.2]{Thiang14} 
to normalise our symmetry involutions so $R_T$ commutes with $R_C$ and $R_T R_C = R_{CT}$.

We start with the full 
group $G=\{1,T,C,CT\}$ and define a left-action on 
$E_A\oplus E_A$  given by left-multiplication by the 
real algebra generated by the elements
$$ 
\left\{ \begin{pmatrix} R_C & 0 \\ 0 & -R_C \end{pmatrix}, 
\begin{pmatrix} 0 & R_C \\ R_C & 0 \end{pmatrix}, 
\begin{pmatrix} 0 & -R_{CT} \\ R_{CT} & 0 \end{pmatrix}\right\},
\quad \Gamma = \begin{pmatrix} \mathrm{sgn}(H) & 0 \\ 0 & \mathrm{sgn}(H) \end{pmatrix}. 
$$
One readily checks as in Proposition \ref{prop:PUA_rep_gives_Clifford_rep} 
that the generating elements have odd grading and mutually anti-commute. 
The left-action generated by these elements gives rise to four distinct 
Clifford algebras depending on whether $R_T^2=\pm 1$ and 
$R_C^2=\pm 1$.

Similarly for the case $G=\{1,C\}$ we take a left-action generated by
$$  
\left\{ \begin{pmatrix} R_C & 0 \\ 0 & -R_C \end{pmatrix}, 
\begin{pmatrix} 0 & R_C \\ R_C & 0 \end{pmatrix}\right\},
\quad \Gamma = \begin{pmatrix} \mathrm{sgn}(H) & 0 \\ 0 & \mathrm{sgn}(H) \end{pmatrix}. 
$$
We obtain an adjointable left-action of either $C\ell_{2,0}$ or $C\ell_{0,2}$ 
depending on whether $R_C^2=\pm 1$.

If $G=\{1,S\}$ then we take the (complex) left-action generated by $R_S$ 
on the complex module $E_{A_\C}$ with grading 
$\mathrm{sgn}(H)$. Hence the left-action is a graded representation of $\C\ell_1$.

Once again the case of $G=\{1,T\}$ is slightly more complicated as 
$R_T$ is evenly graded. If $R_T^2=1$, then $R_T$ implements a 
Real involution on the module $E_A$ and gives no additional 
Clifford representation. If $R_T^2=-1$, then $R_T$ encodes a 
quaternionic structure on $E_A$. There is an equivalence 
between graded quaternionic modules and graded real 
modules with a left $C\ell_{4,0}$-action. Specifically, 
we take $E_A\oplus E_A$ and consider the real action generated by
\begin{align*}
   \left\{\begin{pmatrix} 0 & 1 \\ 1 & 0 \end{pmatrix}, 
   \begin{pmatrix} 0 & -i \\ i & 0 \end{pmatrix}, 
   \begin{pmatrix} 0 & -R_T \\ R_T & 0 \end{pmatrix}, 
   \begin{pmatrix} 0 & -iR_T \\ iR_T & 0 \end{pmatrix}\right\}, 
   \quad \Gamma = \begin{pmatrix} \mathrm{sgn}(H) & 0 \\ 0 & -\mathrm{sgn}(H) \end{pmatrix}.
\end{align*}
We may also replace $i$ with $\begin{pmatrix} 0 & -1 \\ 1 & 0 \end{pmatrix}$ and $iR_T$ with $\begin{pmatrix} R_T & 0 \\ 0 & -R_T \end{pmatrix}$ in order to consider the real action on $E_A^{\oplus 4}$. In either case we obtain a 
graded adjointable representation of $C\ell_{4,0}$.
\end{proof}

\begin{cor} \label{cor:PUA_KK_class}
Let $H$ be a Hamiltonian satisfying Assumption \ref{def:bulk_Hamiltonian_condition} 
that is compatible with a subgroup $G$ of the symmetry group $\{1,T,C,CT\}$. 
Then for $G$ non-trivial and $G\neq \{1,S\}$ the tuple
$$
  \left( C\ell_{n,0}, \, E^{\oplus N}_A, \, 0, \Gamma \right)
$$
is a real Kasparov module with the left action and grading 
given by Proposition \ref{prop:Clifford_rep_on_module}. If 
$G\subset\{1,S\}$ then the Kasparov module is complex.
\end{cor}
\begin{proof}
Because the Dirac-type operator is $0$ and $E_A^{\oplus N}$ is finite projective, the remaining 
conditions required to be a Kasparov module are satisfied.
\end{proof}

Hence given a symmetry compatible gapped Hamiltonian $H$, 
we obtain a $KK$-theory class that encodes the PUA representation 
of $G$ with respect to the grading $\mathrm{sgn}(H)$.

The Clifford representations that we construct in Proposition \ref{prop:Clifford_rep_on_module} are analogous 
to the representations in Proposition \ref{prop:PUA_rep_gives_Clifford_rep} 
and therefore are distinct up to stable isomorphism by 
Table \ref{table:PT_classes_repn_table}. Hence,  as in  the 
Hilbert space picture, there is a $1$-$1$ correspondence 
between symmetry compatible Hamiltonians and graded 
Clifford representations on the real or complex $C^*$-module 
$E_A^{\oplus N}$ (again, up to stable isomorphism).

\begin{remark}[Extensions of our method]
We note that the construction of the crossed product, 
$A\rtimes G$, and $C^*$-module, $E_A$, is 
independent of the finite group under consideration. 
Indeed, we can extend the results of 
Proposition \ref{prop:Clifford_rep_on_module} 
to any finite group $G$ that is compatible with the 
Hamiltonian $H$ in the sense of 
Definition \ref{def:Hamiltonian_compatible_with_G} 
and obtain the Kasparov module 
$\left(C^*(G), E_A, 0, \mathrm{sgn}(H)\right)$, 
which will give a class in $KKO(C^*(G),A)$ or 
$KK(C^*(G)\otimes_\R \C,A_\C)$. One of the key properties 
of a subgroup $G$ of the $CT$-symmetry group is that 
a left-action of $C^*(G)$ or $C^*(G)\otimes_\R \C$ gives rise to a real or complex
Clifford action (using matrices as above), which may not hold for an 
arbitrary finite group $G$. We emphasise 
the flexibility of our method, as it can
accommodate symmetry groups that  contain spatial involution, for example.
\end{remark}

\subsubsection{Projective submodules and $KK$-classes}
Proposition \ref{prop:Clifford_rep_on_module} gives a graded 
representation of Clifford algebras on the $C^*$-module 
coming from the crossed product $A\rtimes G$. Rather than 
use the full $C^*$-module $E_A$, which may be too large, one is often interested 
in projective submodules $P E_A$, where $P$ is some 
projection with even grading.

In the case of $KK$-classes coming from gapped Hamiltonians, 
our obvious choice for a projection is the Fermi projection 
$P_\mu = \chi_{(-\infty,0]}(H)$. In the case of only time-reversal 
symmetry we immediately obtain the following result.

\begin{prop} \label{prop:TR_projective_KK_class}
Let $H$ be a gapped Hamiltonian that is compatible with 
$G\subset \{1,T\}$. Then the tuple
$$
 \left( C\ell_{n,0}, \, P_\mu E_A^{\oplus N}, \, 0, \, \Gamma \right)
$$
is a Kasparov module with $N\in\{1,4\}$ and $n$ determined 
by Table \ref{table:PT_classes_repn_table}
\end{prop}
\begin{proof}
The Fermi projection $P_\mu$ commutes with $R_T$. Hence the 
proof of Corollary \ref{cor:PUA_KK_class} carries over.
\end{proof}

We see that if $G$ is trivial, then the algebras and modules in 
Proposition \ref{prop:TR_projective_KK_class} can be complexified 
to obtain the module $\left( \C, P_\mu A_A, 0, \Gamma\right)$ with 
$\Gamma$ a grading on $A$ and $P_\mu$ degree $0$. If 
$A$ is trivially graded, then this module is exactly the 
representative of the Fermi projection $[P_\mu]\in K_0(A)$ 
translated into $KK(\C, A)$. Hence Proposition 
\ref{prop:TR_projective_KK_class} can be considered as an 
extension of the class of the Fermi projection to systems with 
time reversal symmetry.

Of course we would like an analogue of Proposition 
\ref{prop:TR_projective_KK_class} for when $G$ contains 
charge conjugation or sublattice symmetry. This presents us with 
an issue as
\begin{align*}
  &R_C P_\mu R_C^* = 1-P_\mu,  &&R_S P_\mu R_S^* = 1-P_\mu
\end{align*}
and so the left-action on $E_A$ from Proposition 
\ref{prop:Clifford_rep_on_module} will not descend to the 
projective submodule.

The solution to this issue requires us to consider a different 
grading on the crossed product $A\rtimes G$ and is similar to 
the recent work of Kellendonk~\cite{Kellendonk15}. We assume 
$G=\{1,S\}$, $\{1,C\}$ or $\{1,C,T,CT\}$ and define a grading on 
$E_A$ by $\mathrm{Ad}_{R_S}$ or $\mathrm{Ad}_{R_C}$. The 
Hamiltonian $H\in A\rtimes G$ is now odd with respect to this 
grading and the operators $R_C$ and $R_S$ are now even.

Next we consider the new $C^*$-module given by the graded 
tensor product $(E\hat\otimes C\ell_{0,1})_{A\hat\otimes C\ell_{0,1}}$. 
The right-action of $C\ell_{0,1}$ is given by right-multiplication 
and the product space has inner product
$$
  ( e_1 \hat\otimes \nu_1\mid e_2\hat\otimes \nu_2)_{A\hat\otimes C\ell_{0,1}} 
    = \Phi( e_1^* e_2) \hat\otimes \nu_1^*\nu_2.
$$
Inside of $E\hat\otimes C\ell_{0,1}$ is the element 
$\tilde{H} = H \hat\otimes \rho$, where $\rho$ is the odd
generator of $C\ell_{0,1}$ that is skew-adjoint and squares 
to $-1$. Because both $H$ and $\rho$ are odd, $\tilde{H}$ is 
even and by the properties of graded tensor products 
$\tilde{H}^* = (H\hat\otimes \rho)^* = -H \hat\otimes (-\rho) = \tilde{H}$. 
Therefore $\tilde{H}$ is self-adjoint and invertible with 
inverse $H^{-1}\hat\otimes \rho$. Thus 
$\tilde{P}_\mu = \chi_{(-\infty,0]}(\tilde{H})$ is an 
even projection in $(A\rtimes G)\hat\otimes C\ell_{0,1}$ and we 
have the following result.

\begin{prop} \label{prop:projective_KK_class_odd}
Let $H$ be a gapped Hamiltonian compatible with
$G=\{1,C\}$ or $\{1,C,T,CT\}$. Then
$$
 \left( C\ell_{n,1},\, \tilde{P}_\mu(E \hat\otimes C\ell_{0,1})^{\oplus 4}_{A\hat\otimes C\ell_{0,1}}, \, 0, \,\mathrm{Ad}_{R_C} \hat\otimes \gamma_{C\ell_{0,1}} \right)
$$
is a real Kasparov module with $n$ given in
Table \ref{table:PT_classes_repn_table}. If $G=\{1,S\}$, then
$$
\left( \C\ell_2, \,  \tilde{P}_\mu(E \hat\otimes \C\ell_{1})^{\oplus 2}_{A_\C\hat\otimes \C\ell_{1}}, \, 0, \,\mathrm{Ad}_{R_S} \hat\otimes \gamma_{\C\ell_{1}} \right)
$$
is a complex Kasparov module.
\end{prop}
\begin{proof}
We first take the operators $R_C$ and $R_S$ to be commuting
by~\cite[Proposition 6.2]{Thiang14} with $R_C$ self-adjoint
(resp. skew-adjoint) if $R_C^2=1$ (resp. $R_C^2=-1$). We
can impose the same condition on $R_T$, which determines
the behaviour of $R_S=R_{CT}$.
If $G=\{1,S\}$ we may take $R_S^2=1$ and $R_S^*=R_S$.

Next we consider the operators
\begin{align*}
   &R_C\hat\otimes \rho,  &&R_S \hat\otimes \rho,
\end{align*}
which are odd in $(E\hat\otimes C\ell_{0,1})_{A\hat\otimes C\ell_{0,1}}$.
The new operators 
are self-adjoint or skew adjoint depending on whether
$(R_C\hat\otimes \rho)^2 = \pm 1$ (similarly $R_S\hat\otimes \rho$).
Hence the operators $R_C\hat\otimes \rho$ and
$R_{CT}\hat\otimes \rho$ act as generators of a Clifford
algebra on $(E\hat\otimes C\ell_{0,1})_{A\hat\otimes C\ell_{0,1}}$.
Furthermore, we check that
\begin{align*}
  (R_C\hat\otimes \rho)(H\hat\otimes \rho)(R_C\hat\otimes \rho)^*
   &= (-R_C H \hat\otimes (-1))(-R_C^* \hat\otimes \rho) \\
   &= -R_C H R_C^* \hat\otimes \rho = H\hat\otimes \rho.
\end{align*}
Hence $R_C\hat\otimes \rho$ commutes with $\tilde{H}$
and so is an adjointable operator on the projective submodule
$\tilde{P}_\mu(E \hat\otimes C\ell_{0,1})$ (similarly
$R_S \hat\otimes \rho$).

For the group $\{1,C,T,CT\}$, our Clifford action
changes depending on whether $R_C^2=R_T^2$ or $R_C^2 = -R_T^2$. We
refer the reader to~\cite{Kellendonk15} for a more detailed
exposition of these subtle differences and how they arise.
For the case $R_C^2 = R_T^2 = \pm 1$, we define the left-action on
$\tilde{P}_\mu(E \hat\otimes C\ell_{0,1})^{\oplus 4}$ generated
by
\begin{align*}
  &\left\{ \begin{pmatrix} (R_{CT}\hat\otimes \rho)\otimes 1_2 & 0_2 \\ 0_2 & -(R_{CT}\hat\otimes \rho)\otimes 1_2 \end{pmatrix},
    \begin{pmatrix} 0 & 0 & 0 & -R_{CT}\hat\otimes \rho \\ 0 & 0 & R_{CT}\hat\otimes \rho & 0 \\ 0 & R_{CT}\hat\otimes \rho & 0 & 0 \\ -R_{CT}\hat\otimes \rho & 0 & 0 & 0 \end{pmatrix},  \right. \\
  &\qquad \left. \begin{pmatrix} 0 & 0 & R_C\hat\otimes\rho & 0 \\ 0 & 0 & 0 & -R_C\hat\otimes\rho \\ -R_C\hat\otimes\rho & 0 & 0 & 0 \\ 0 & R_C\hat\otimes\rho & 0 & 0  \end{pmatrix},
    \begin{pmatrix} 0 & 0 & 0 & R_C\hat\otimes\rho \\ 0 & 0 & R_C\hat\otimes\rho & 0 \\ 0 & -R_C\hat\otimes\rho & 0 & 0 \\ -R_C\hat\otimes\rho & 0 & 0 & 0 \end{pmatrix} \right\}.
\end{align*}
A careful check shows that the generators 
mutually anti-commute and are odd under the grading
$(\mathrm{Ad}_{R_C}\hat\otimes \gamma_{C\ell_{0,1}})^{\oplus 4}$.
Two generators square to
$-R_{CT}^2\otimes 1_4$, and two generators square to $R_C^2\otimes 1_4$.
Thus the left action gives a representation of a real
Clifford algebra $C\ell_{2,2}$ or $C\ell_{0,4}$ on the projective module
determined by the sign of $R_C^2$ and $R_{T}^2$. As the Dirac-type operator
is $0$, we obtain a real Kasparov module.

Next we consider the full symmetry group $G=\{1,C,T,CT\}$ with twisted representation
such that $R_C^2 = -R_T^2$. We consider a left-action with generating
elements
\begin{align*}
  &\left\{ \begin{pmatrix} 0 & R_{CT}\hat\otimes \rho & 0 & 0 \\ -R_{CT}\hat\otimes \rho & 0 & 0 & 0 \\ 0 & 0 & 0 & -R_{CT}\hat\otimes \rho \\ 0 & 0 & R_{CT}\hat\otimes \rho & 0 \end{pmatrix},
    \begin{pmatrix} 0 & 0 & 0 & -R_{CT}\hat\otimes \rho \\ 0 & 0 & R_{CT}\hat\otimes \rho & 0 \\ 0 & R_{CT}\hat\otimes \rho & 0 & 0 \\ -R_{CT}\hat\otimes \rho & 0 & 0 & 0 \end{pmatrix},  \right. \\
  &\qquad \qquad \qquad \left.
  \begin{pmatrix} 0 & 0 & R_C\hat\otimes\rho & 0 \\ 0 & 0 & 0 & -R_C\hat\otimes\rho \\ -R_C\hat\otimes\rho & 0 & 0 & 0 \\ 0 & R_C\hat\otimes\rho & 0 & 0  \end{pmatrix},
    \begin{pmatrix} 0 & 0 & 0 & R_C\hat\otimes\rho \\ 0 & 0 & R_C\hat\otimes\rho & 0 \\ 0 & -R_C\hat\otimes\rho & 0 & 0 \\ -R_C\hat\otimes\rho & 0 & 0 & 0 \end{pmatrix} \right\}
\end{align*}
and grading $(\mathrm{Ad}_{R_{CT}}\hat\otimes \gamma_{C\ell_{0,1}})^{\oplus 4}$.
The generators give rise to an action of
$C\ell_{3,1}$ or $C\ell_{1,3}$ depending on the sign of $R_C^2$ with $R_C^2=-R_T^2$.

If $G=\{1,S\}$, then we take the following generators,
$$
\left\{ \begin{pmatrix} R_{S}\hat\otimes \rho & 0 \\ 0 & -R_{S}\hat\otimes \rho \end{pmatrix},
 \begin{pmatrix} 0 & R_S\hat\otimes \rho \\ R_S\hat\otimes \rho & 0 \end{pmatrix} \right\},
 \qquad \Gamma = \begin{pmatrix} \mathrm{Ad}_{R_S}\hat\otimes \gamma_{\C\ell_{1}} & 0 \\ 0 & \mathrm{Ad}_{R_S}\hat\otimes \gamma_{\C\ell_{1}} \end{pmatrix},
$$
which gives a $\C\ell_2$-action on $\tilde{P}_\mu(E \hat\otimes \C\ell_1)^{\oplus 2}$.

Finally if $G=\{1,C\}$, our Clifford generators are
\begin{align*}
  &\left\{ \begin{pmatrix} 0 & 0 & R_C\hat\otimes\rho & 0 \\ 0 & 0 & 0 & -R_C\hat\otimes\rho \\ -R_C\hat\otimes\rho & 0 & 0 & 0 \\ 0 & R_C\hat\otimes\rho & 0 & 0  \end{pmatrix},
    \begin{pmatrix} 0 & 0 & 0 & R_C\hat\otimes\rho \\ 0 & 0 & R_C\hat\otimes\rho & 0 \\ 0 & -R_C\hat\otimes\rho & 0 & 0 \\ -R_C\hat\otimes\rho & 0 & 0 & 0 \end{pmatrix},
   \right. \\
   &\qquad \qquad \qquad \left. \begin{pmatrix} (u\hat\otimes \rho)\otimes 1_2 & 0_2 \\ 0_2 & -(u\hat\otimes \rho)\otimes 1_2 \end{pmatrix}   \right\},
     \quad \Gamma = (\mathrm{Ad}_{R_C}\hat\otimes C\ell_{0,1})^{\oplus 4}.
\end{align*}
where $u$ is an even self-adjoint unitary in $A\rtimes G$ that anti-commutes
with $H$ (passing to stabilisation/matrices if necessary). The left-action
generates $C\ell_{2,1}$ or $C\ell_{0,3}$ depending on the sign of $R_C^2$.

Taking the Clifford actions up to stable isomorphism, we obtain a left-action
of $C\ell_{n,1}$ or $\C\ell_{n+1}$ on
$\tilde{P}_\mu(E_A\hat\otimes C\ell_{0,1})_{A\hat\otimes C\ell_{0,1}}^{\oplus N}$
(or the complex $C^*$-module)
with $n$ determined by Table \ref{table:PT_classes_repn_table}.
\end{proof}

If $A$ is trivially graded, we can relate the class of the 
Kasparov modules considered in Propostion \ref{prop:projective_KK_class_odd} 
to $K$-theory by the identification
$$
  KKO(C\ell_{n,0}\hat\otimes C\ell_{0,1}, A\hat\otimes C\ell_{0,1}) \cong KKO(C\ell_{n,0}\hat\otimes C\ell_{1,1}, A) 
   \cong KO_n(A),
$$
where we have used stability of $KKO$ and Proposition \ref{prop:real_Real_k_with_kasparov_equivalence} (similarly 
complex $K$-theory).

We shall denote the class of the projective Kasparov modules 
constructed in Proposition \ref{prop:TR_projective_KK_class} 
and \ref{prop:projective_KK_class_odd}  
in $ KKO(C\ell_{n,0},A)$ (or in the complex case $KK(\C\ell_n,A_\C)$) by $[H^G]$.

For trivially graded algebras, the class $[H^G]$ 
defines a class in either 
$KO_n(A)$ or $K_n(A_\phi^d)$. 
Indeed for $A= M_N(C^*(\Z^d))$, we have that
$$ 
KKO(C\ell_{n,0},M_N(C^*(\Z^d))) 
\cong KO_n(C^*(\Z^d)) 
\cong KO^{-n}(\T^d), 
$$
where we have used stability and Proposition \ref{prop:real_Real_k_with_kasparov_equivalence}.
Hence we recover 
Kitaev's classification for models with discrete translation symmetry~\cite{Kitaev09}, 
though we note that the noncommutative method allows for more 
complicated algebras and spaces to be considered. We can use the Pimsner-Voiculescu 
sequence with trivial action to find that 
$KO_n(C^*(\Z^d))\cong KO_n(C^*(\Z^{d-1}))\oplus KO_{n-1}(C^*(\Z^{d-1}))$ 
and therefore
$$ 
KO_n(C^*(\Z^d)) \cong \bigoplus_{k=0}^d \binom{d}{k} KO_{n-k}(\R), 
$$
see~\cite{PV80, Thiang14}.

\begin{remark}[Anti-linear symmetries and Real $C^*$-algebras]
We have shown how the symmetries coming from the group $\{1,T,C,CT\}$ can be linked to real $C^*$-algebras and
$KKO$-theory. One may ask whether we can also study this question from the perspective of Real
$C^*$-algebras and $KKR$-theory. The construction of the crossed product $A\rtimes G$ where $\alpha_g(a) = R_g a
R_g^*$ will \emph{not} hold in the Real category if $G=\{1,T,C,CT\}$ as this will involve two anti-linear
automorphisms $\alpha_C$ and $\alpha_T$. However, if we consider the subgroups $\{1,T\}$ or $\{1,C\}$ with $R_T$
or $R_C$ defining a Real structure on the (complex) Hilbert space $\calH$, 
then $\alpha_T(a) = R_T aR_T^*$ (or $\alpha_C(a)=R_CaR_C^*$) defines a Real involution $a\mapsto a^\tau$ on the complex algebra 
$A\otimes_\R \C \subset M_N(C^*(\Z^d)\otimes_\R\C)$ with $a^\tau = \alpha_h(a)$ for $h=T$ or $C$.

We expect similar results
to hold in the Real picture provided $G=\{1,T\}$ or $\{1,C\}$. 
In the interest of brevity, we will leave a proper investigation 
of the wider links between insulator systems and $KKR$-theory 
to another place.
\end{remark}

\subsection{Spectral triples and pairings}  \label{subsec:spec_trip_and_pairings}
Our discussion up to this point has centred on the connection of symmetries
with $KKO$-theory, but this is not the end of the story. 
Recall from Section \ref{subsec:IQHE} that one also 
obtains topological information coming from the geometry of the (possibly noncommutative) 
Brillouin zone.

For complex discrete systems without disorder, a Hamiltonian $H$ that satisfies Assumption \ref{def:bulk_Hamiltonian_condition} is contained in $M_N(A_\phi^d)$. We can consider a  dense $\ast$-subalgebra $\calA_\C$ of finite polynomials of (twisted) shift operators and construct the complex spectral triple
\begin{equation} \label{eq:complex_insulator_spec_trip}
   \left( \calA_\C,\, \ell^2(\Z^d)\otimes \C^N\otimes \C^\nu,\, \sum_{j=1}^d X_j\otimes 1_N\otimes \gamma^j,\, \gamma = (-i)^{d/2}\gamma^1\cdots\gamma^d\right), 
\end{equation}
where the matrices $\gamma^j$ have the relation $\gamma^i \gamma^j + \gamma^j\gamma^i = 2\delta_{i,j}$. In~\cite{PLB13,PSB14}, one obtains `higher order Chern numbers' by taking the index pairing of this spectral triple with the Fermi projection or some unitary $u\in\calA$. 

 For $d$ even the index pairing is given by the map
\begin{align*}
     KK(\C, A_\C) \times KK(A_\C,\C) \to KK(\C,\C) \cong \Z. 
\end{align*}
Recall that   $KK(\C, A_\C)$ is isomorphic to the $K$-theory of the algebra $A$
while elements of  $KK(A_\C,\C)$ are represented by spectral triples and the integer 
constructed by the pairing in this instance
is a  Fredholm index.

Specialising now to the case in hand we can form a product:
\begin{align*}
    C_d &= [P_\mu] \hat\otimes_{A_\phi^d} \left[\left( \calA_\C, \ell^2(\Z^d)\otimes \C^N\otimes \C^\nu, X=\sum_{j=1}^d X_j\otimes 1_N\otimes \gamma^j, \gamma \right) \right] \\
       &= \Index(P_\mu X_+ P_\mu),
\end{align*}
where $X = \begin{pmatrix} 0 & X_- \\ X_+ & 0 \end{pmatrix}$ is decomposed by the grading $\gamma$. The case of $d$ odd has an analogous formula except that in this instance we are taking a product of $KK(\C\ell_1,A_\C)$ with $KK(A_\C,\C\ell_1)$. Our goal below is to refine this complex index pairing so that it applies to the real picture. This is of course essential when one considers time-reversal and charge-conjugation symmetry.

\subsubsection{The bulk spectral triple}  \label{sec:KKO_fundamental_class}
Spectral triples over real algebras require representations 
on real Hilbert spaces. Hence, we take $A\subset M_N(C^*(\Z^d))$ 
acting on the bulk Hilbert space $\calH_b$, which can be 
$\ell^2(\Z^d) \otimes \R^N$ or $\ell^2(\Z^d)\otimes \C^{M}$, 
where $\C\cong \R\oplus i\R$ is considered as a real space. 

If a Hamiltonian $H$ satisfies Assumption \ref{def:bulk_Hamiltonian_condition}, 
then we take $\calA$ to be the $\ast$-algebra of finite 
polynomials of matrices of shift operators (or infinite polynomials 
with Schwartz-class coefficients) over $\R$. Such an algebra 
$\calA$ is dense in $A$. We require 
a dense subalgebra in order to deal with the 
commutator condition in spectral triples and also 
for unbounded Kasparov modules over $C^*(\Z^d)$. 
Using constructions similar  to~\cite{KasparovNovikov, LRV12} 
for the Hodge-de Rham spectral triple, we have the following result.
\begin{prop} 
\label{prop:real_bulk_spec_trip}
If a Hamiltonian $H$ satisfies 
Assumption \ref{def:bulk_Hamiltonian_condition} with 
$\calA \subset M_N(C^*(\Z^d))$, then
$$ 
\lambda =\left(\calA \hat\otimes C\ell_{0,d},\, \calH_b \otimes \bigwedge\nolimits^{\!*} \R^d,\, 
\calD =\sum_{j=1}^d X_j \otimes  \gamma^j,\, \gamma_{\bigwedge^* \R^d} \right) 
$$
is a real spectral triple, where $X_j$ is the position 
operator on $\ell^2(\Z^d)$ and acts diagonally on 
$\calH_b$. The left-action of $C\ell_{0,d}$ is 
generated by the operators $\{\rho^j\}_{j=1}^d$ 
and the operators $\{\gamma^j\}_{j=1}^d$ generate 
$C\ell_{d,0}$. The Clifford algebras $C\ell_{0,d}$ and 
$C\ell_{d,0}$ are represented as left and right actions on 
$\bigwedge^* \R^d$ respectively by the formulae
\begin{align}  \label{eq:left_right_real_clifford_actions_on_exterior_algebra}
   &\rho^j(\omega) = e_j\wedge \omega - \iota(e_j)\omega,   
   &&\gamma^j(\omega) = e_j \wedge \omega + \iota(e_j)\omega,
\end{align}
with $\omega\in\bigwedge^*\R^d$, $\{e_j\}_{j=1}^d$ 
the standard basis of $\R^d$ and $\iota(v)\omega$ 
the contraction of $\omega$ along $v$. The grading 
$\gamma_{\bigwedge^* \R^d}$ is given in terms of the isomorphism 
$C\ell_{0,d}\hat\otimes C\ell_{d,0} \cong \End_\R(\bigwedge^* \R^d)$, 
where 
$\gamma_{\bigwedge^* \R^d} = (-1)^d \rho^1\cdots\rho^d\hat\otimes \gamma^d\cdots\gamma^1$.
\end{prop}

One can check that $\rho^j$ and $\gamma^k$ 
anti-commute (i.e. they graded-commute). 
We note that, despite a right-action by 
$C\ell_{d,0}$ on $\bigwedge^*\R^d$, we do \emph{not} 
get an $A\hat\otimes C\ell_{0,d}$-$C\ell_{d,0}$ 
Kasparov module as the graded-commutator of 
$1\otimes\gamma^k$ with $\sum_j X_j\otimes \gamma^j$ 
is not bounded (see~\cite[Section 4.3]{LRV12} for a 
more detailed discussion on these Clifford actions and the link to Kasparov's fundamental class).  

\begin{proof}[Proof of Proposition \ref{prop:real_bulk_spec_trip}]
The generators $\rho^j$ of the left action of $C\ell_{0,d}$ 
graded-commute with $\sum_j X_j\otimes\gamma^j$ 
so we just need to check that $[X_j\otimes 1_N,a]$ is 
bounded for all $j$ and $a(1+\calD^2)^{-1/2}$ is compact for 
$a\in\calA$. We let 
${S}^\alpha = {S}_1^{\alpha_1}\cdots {S}_d^{\alpha_d}$ 
for $S_j$ a shift operator and $\alpha = (\alpha_1,\ldots,\alpha_d)\in\Z^d$. Then, for 
$\psi\in\Dom(X_j)\subset \ell^2(\Z^d)$,
\begin{align*}
  [X_j,{S}^\alpha]\psi(x) 
  &= x_j  \psi(x-\alpha) -  (x_j-\alpha_j)\psi(x-\alpha) \\
    &= \alpha_j ({S}^\alpha\psi)(x),
\end{align*}
Therefore $[X_j,a]$ extends to a bounded operator for 
$a$ any finite polynomial of ${S}^\alpha$ on $\ell^2(\Z^d)$. 
Because matrices of such elements generate $\calA$, 
$[X_j\otimes 1_N,a]$ is bounded for any $a\in\calA$.

Next we note that 
$(1+\calD^2)^{-1/2} = (1+|X|^2)^{-1/2}\otimes 1_N \otimes 1_{\bigwedge^*\R^d}$ 
as an operator on  on 
$\calH_b = \ell^2(\Z^d)\otimes \mathbb{F}^N\otimes \bigwedge^*\R^d$ 
for $\mathbb{F}=\R$ or $\C$. On $\ell^2(\Z^d)$,
$$ 
(1+|X|^2)^{-1/2} = \bigoplus_{k\in\Z^d} (1+|k|^2)^{-1/2}P_k, 
$$
where $P_k$ is the projection onto the span of $e_{(k_1,\dots,k_d)}$ 
with $\{e_k\}_{k\in\Z^d}$ the standard basis of $\ell^2(\Z^d)$. Hence 
$(1+|X|^2)^{-1/2}$ is a norm-convergent sum of finite-rank operators 
and so is compact. From this we conclude that $(1+D^2)^{-1/2}$ is 
compact on $\calH_b\otimes\bigwedge^*\R^d$.
\end{proof}

\begin{remark}
The spectral triple from Proposition \ref{prop:real_bulk_spec_trip} 
uses the oriented structure on the (noncommutative) $d$-torus to 
construct the Clifford actions and Dirac-type operator. 
We note that a spectral triple can also be built using the 
spin structure (or spin$^c$ structure) on the torus. The spin spectral triple is the 
same as the spectral triple from Proposition \ref{prop:real_bulk_spec_trip} 
up to a Morita equivalence bimodule. We use the oriented structure to obtain 
explicit representations of the Clifford generators and actions. See \cite{KasparovNovikov,LRV12},
where the relationship of the spectral triple $\lambda$ to the fundamental class (in the sense of
Poincar\'e duality in $KK$-theory) of the
Brillouin zone is made clear. This choice of spectral triple ensures, 
at least in the absence of disorder, that all Clifford module valued index pairings can be 
faithfully detected.
\end{remark}

For the case of complex algebras, the spectral triple of interest is given in Equation \eqref{eq:complex_insulator_spec_trip}.
We think of the real spectral triple of Proposition \ref{prop:real_bulk_spec_trip} as encoding geometric information of the (possibly noncommutative) Brillouin torus, including dimension. The Kasparov module represented by $[H^G]$ on the other hand captures information about the internal symmetries of the Hamiltonian. By taking the pairing/product of the 
class $[H^G]$ with the spectral triple, we obtain measurable quantities which
reflect the topological properties of the system.

\begin{remark}[Pairings and the periodic table]
The construction above of an unbounded Kasparov module gives a class $[\lambda]\in KKO(A\hat\otimes C\ell_{0,d},\R)$~\cite{BJ83}. We would like to consider an analogous notion in real Kasparov theory of the Chern numbers. However, because we are dealing with representatives of $KKO$-classes, we need to generalise the complex pairing to the internal product of $[\lambda]$ with the class $[H^G]$ from Proposition \ref{prop:TR_projective_KK_class} 
and \ref{prop:projective_KK_class_odd}  that represents the 
symmetries of the Hamiltonian. That is we take the
Kasparov product, a well-defined map
\begin{align*}
  &KKO(C\ell_{n,0},A) \times KKO(A\hat\otimes C\ell_{0,d},\R) \to KKO(C\ell_{n,0}\hat\otimes C\ell_{0,d},\R) 
\end{align*}
and this leads to a Clifford module valued index
\begin{align*}
  &C_{n,d} = [H^G] \hat\otimes_A [\lambda].
\end{align*}
We have to represent the index pairing as a Kasparov product rather than a pairing of a projection with a cyclic cocyle as the latter involves a map to periodic cyclic cohomology, which is unable to detect torsion invariants. We note that the class $[H^G]\hat\otimes_A [\lambda]$ takes values in $KKO(C\ell_{n,0}\hat\otimes C\ell_{0,d},\R) \cong KO_{n-d}(\R)$~\cite[\S{6}]{Kasparov80}. Therefore, by considering the various symmetry subgroups of $\{1,T,C,CT\}$ that give rise to graded Clifford representations of $C\ell_{n,0}$ for different $n$ outlined in Table \ref{table:PT_classes_repn_table}, we are able to derive the celebrated periodic table of strong topological phases, which is given in Table \ref{table:Periodic_table_up_to_3d}.
\end{remark}

We summarise the discussion  in this subsection in the following result.
\begin{prop}
The periodic table of (strong) topological phases can be realised as the index pairing (Kasparov product) of the real/complex Kasparov module $[H^G]$ from 
Proposition \ref{prop:TR_projective_KK_class} 
and \ref{prop:projective_KK_class_odd} with the bulk spectral triple of Proposition \ref{prop:real_bulk_spec_trip} or Equation \eqref{eq:complex_insulator_spec_trip}.
\end{prop}

\begin{remark}[The even-integer index]
We note that $KO_4(\R)\cong \Z$ whereas the 
pairings $[H^G]\hat\otimes_A [\lambda]$ that take 
value in $KO_4(\R)$ can be associated to an even integer. 
The relationship is 
that we associate $[H^G]\hat\otimes_A [\lambda]$ to a Clifford module, 
which in this case can be expressed as a quaternionic Fredholm index. 
Since both the kernel and cokernel of the relevant operator is a quaternionic
space, the complex dimension of these spaces 
take even values, and so we can express the index as an even integer. 
See Section \ref{subsec:KKO_Clifford_index} 
and~\cite[p23]{AS69}. See also~\cite[Section 6.2]{DNSB14b} for some of 
the physical implications of an even-valued index.
\end{remark}

\begin{table} 
    \centering
    \begin{tabular}{ p{1.7cm} | c c | p{2.5cm} | c c c c }
       \multirow{2}{1.7cm}{Symmetry generators} & \multirow{2}{*}{$R_C^2$} & \multirow{2}{*}{$R_T^2$} &     \multirow{2}{2.5cm}{Graded Representation} & \multicolumn{4}{c}{$[H^G]\hat\otimes [\lambda]\in KO_{n-d}(\R)$ or $K_{n-d}(\C)$} \\ 
         &  &  &  & $d=0$ & $d=1$ & $d=2$ & $d=3$ \\ \hline
       $T$ &  & $+1$ & $C\ell_{0,0}$ & $\Z$ & 0 & 0 & 0 \\
       $C,T$ & $+1$ & $+1$ & $C\ell_{1,0}$ & $\Z_2$ & $\Z$ & 0 & 0 \\
       $C$ & $+1$ &  & $C\ell_{2,0}$ & $\Z_2$ & $\Z_2$ & $\Z$ & 0 \\
       $C,T$ & $+1$ & $-1$ & $C\ell_{3,0}$ & 0 & $\Z_2$ & $\Z_2$ & $\Z$ \\
       $T$ &  & $-1$ & $C\ell_{4,0}$ & $(2)\Z$ & 0 & $\Z_2$ & $\Z_2$ \\
       $C,T$ & $-1$ & $-1$ & $C\ell_{5,0}$ & 0 & $(2)\Z$ & 0 & $\Z_2$ \\
       $C$ & $-1$ &  & $C\ell_{6,0}$ & 0 & 0 & $(2)\Z$ & 0 \\
       $C,T$ & $-1$ & $+1$ & $C\ell_{7,0}$ & 0 & 0 & 0 & $(2)\Z$ \\ \hline \hline 
       N/A &  &  &  $\C\ell_{0}$ & $\Z$ & 0 & $\Z$ & 0 \\
       $S$ & \multicolumn{2}{c|}{$R_S^2=1$}  & $\C\ell_1$ & 0 & $\Z$ & 0 & $\Z$
    \end{tabular}
    \caption{Symmetry types, their corresponding graded Clifford representation and the pairing of the symmetry class with the $d$-dimensional spectral triple (shown for $d\leq 3$). \label{table:Periodic_table_up_to_3d}}
\end{table}

\subsection{The Kasparov product and the Clifford index} \label{subsec:KKO_Clifford_index}
So far we have identified the invariants of interest in topological insulator systems as a Kasparov product, $[H^G]\hat\otimes_A [\lambda]$, of Kasparov modules capturing internal symmetries and geometric information. In the case of complex algebras and modules, this abstract pairing can be concretely represented as a Fredholm index and takes the form $\Index(P X_+ P)$ or $\Index(PuP)$ (for $u$ a unitary) depending on whether $d$ is even or odd. A replacement for this numerical index in the real case
requires the viewpoint of \cite{ABS64,SpinGeometry} in order to express the Kasparov product $[H^G]\hat\otimes_A [\lambda]$  more concretely. The Clifford module interpretation of the index due to Atiyah-Bott-Shapiro  in $KO$-theory is, in fact, essential here.

In order to draw this link, we first must compute the (unbounded) product 
$[H^G]\hat\otimes_A [\lambda]$. The computation will change slightly depending on 
whether $G\subset \{1,T\}$ or if $G$ contains odd symmetries.

\begin{lemma} \label{lemma:real_index_pairing_product}
If $G$ contains the odd symmetries $C$ or $S$, then the real Kasparov product $[H^G]\hat\otimes_A [\lambda]$ 
can be represented by the unbounded Kasparov module
$$ 
\left(C\ell_{n,1}\hat\otimes C\ell_{0,d},\, \tilde{P}_\mu (E^{\oplus N} \otimes_A \calH_b \hat\otimes C\ell_{0,1})_{C\ell_{0,1}} \hat\otimes \bigwedge\nolimits^{\!*} \R^d,\, \sum_{j=1}^d \tilde{P}_\mu (1\otimes_\nabla X_j\hat\otimes 1)\tilde{P}_\mu \hat\otimes\gamma^j \right)
 $$
with grading $(\mathrm{Ad}_{R}\hat\otimes 1\hat\otimes \gamma_{C\ell_{0,1}})\hat\otimes \gamma_{\bigwedge^* \R^d}$.
If $G\subset\{1,T\}$, then $[H^G]\hat\otimes_A [\lambda]$ is represented by the spectral triple
$$
\left(C\ell_{n,0}\hat\otimes C\ell_{0,d},\, P_\mu (E^{\oplus N} \otimes_A \calH_b ) \hat\otimes \bigwedge\nolimits^{\!*} \R^d,\, \sum_{j=1}^d P_\mu (1\otimes_\nabla X_j)P_\mu \hat\otimes\gamma^j,\, (\Gamma\hat\otimes 1)\hat\otimes \gamma_{\bigwedge^* \R^d} \right).
$$
\end{lemma}
\begin{proof}
We will focus on the case of $C$ or $S\in G$ as the case of $G\subset\{1,T\}$ follows by an 
entirely analogous argument but without the extra $C\ell_{0,1}$ actions and modules.

In order to take the product 
$$ KKO(C\ell_{n,1},A\hat\otimes C\ell_{0,1})\times KKO(A\hat\otimes C\ell_{0,d},\R)\to KKO(C\ell_{n,1}\hat\otimes C\ell_{0,d},C\ell_{0,1}), $$  
we first need to take an external product with an identity class in $KKO(C\ell_{0,d},C\ell_{0,d})$ on 
the left and an identity class in $KKO(C\ell_{0,1},C\ell_{0,1})$ on the right. 
The class in $KKO(C\ell_{0,d},C\ell_{0,d})$ can be represented by the Kasparov module
$$ \left( C\ell_{0,d}, \left(C\ell_{0,d}\right)_{C\ell_{0,d}},0,\gamma_{C\ell_{0,d}}\right) $$
with right and left actions given by right and left Clifford multiplication. The 
class of the identity in $KKO(C\ell_{0,1},C\ell_{0,1})$ is represented 
analogously. At the level of $C^*$-modules, the product module is given by
\begin{align*}
  &\left(\tilde{P}_\mu(E^{\oplus N}_A \hat\otimes {C\ell_{0,1}}_{C\ell_{0,1}}) \hat\otimes_\R C\ell_{0,d}\right) \hat\otimes_{A\hat\otimes C\ell_{0,d+1}} \left(\calH_b \hat\otimes {C\ell_{0,1}}_{C\ell_{0,1}}\hat\otimes \bigwedge\nolimits^{\!*} \R^d\right)   \\
   &\hspace{2cm} \cong   \tilde{P}_\mu\left(E^{\oplus N}\otimes_A \calH_b \hat\otimes {C\ell_{0,1}}_{C\ell_{0,1}}\right) \hat\otimes_\R \left(C\ell_{d,0}\cdot \bigwedge\nolimits^{\!*} \R^d\right) \\
    &\hspace{2cm} \cong \tilde{P}_\mu\left(E^{\oplus N}\otimes_A \calH_b \hat\otimes {C\ell_{0,1}}_{C\ell_{0,1}}\right) \hat\otimes_\R \bigwedge\nolimits^{\!*} \R^d
\end{align*}
as the action of $C\ell_{0,d}$ on $\bigwedge^*\R^d$ is bijective. The projection 
$\tilde{P}_\mu$ is defined on $E^{\oplus N}\otimes_A \calH_b \hat\otimes {C\ell_{0,1}}$ by 
the extended $\tilde{H} = H\hat\otimes 1\hat\otimes \rho$, which is still self-adjoint 
and gapped. Furthermore, the action of $C\ell_{n,1}$ and $C\ell_{0,d}$ on 
$\tilde{P}_\mu(E^{\oplus N}_A \hat\otimes {C\ell_{0,1}})$ and $\bigwedge^*\R^d$ 
respectively can be extended to an 
action of $C\ell_{n,1}\hat\otimes C\ell_{0,d}$ on 
$\tilde{P}_\mu (E^{\oplus N}\otimes_A \calH_b \hat\otimes {C\ell_{0,1}}) \hat\otimes \bigwedge\nolimits^{\!*} \R^d$.

Next we construct the operator $1\otimes_\nabla X_j$ on $E^{\oplus N} \otimes_A \calH_b$ 
for $j\in\{1,\ldots,d\}$. First let 
$\calE_{A}$ be the submodule of $E$, which is spanned by elements of the form
$$ \sum_{g\in G} a_g R_g = \sum_{g\in G} R_g \alpha_g^{-1}(a_g) = \sum_{g\in G} R_g \tilde{a}_g. $$
Over finite sums of such elements, we take the connection
$$  
\nabla :\calE \to \calE \otimes_{\mathrm{poly}(a)} \Omega^1(\mathrm{poly}(a)), \qquad  \nabla\left( \sum_{g\in G} R_g a_g\right) = \sum_{g\in G} R_g \otimes \delta(a_g),
$$
where $\delta$ is the universal derivation. We represent $1$-forms on $\calH_b$ via
$$  
\tilde{\pi}\!\left(a_0\delta(a_1)\right)\lambda = a_0[X_j,a_1]\lambda, \quad \lambda\in\calH_b,  
$$
from which we define, for $(e\otimes \lambda) \in E \otimes_A \calH_b$, 
$$  
(1\otimes_\nabla X_j)(e\otimes\lambda) 
:= (e\otimes X_j\lambda) + (1\otimes \tilde{\pi})\circ(\nabla\otimes 1)(e\otimes \lambda). 
$$
We use a connection to correct the naive formula $1\otimes X_j$ is because 
$1\otimes X_j$ is not well-defined on the balanced tensor product. Computing yields that
\begin{align*}
   (1\otimes_\nabla X_j)\left(\sum_{g\in G} R_g a_g \otimes \lambda\right) 
   &= \sum_{g\in G}R_g \otimes a_g X_j \lambda 
   + \sum_{g\in G}R_g \otimes [X_j,a_g]\lambda \\
   &= \sum_{g\in G} R_g \otimes X_j a_g \lambda.
\end{align*}
For the case of $E^{\oplus N} \otimes_A \calH_b$ with $N\geq 2$, we can always inflate 
$\calH_b$ to $\calH_b^{\oplus N}$ and define the operator 
$(1\otimes_\nabla X_j)$ diagonally. Hence we can define the operator  
$\sum_{j=1}^d \tilde{P}_\mu (1\otimes_\nabla X_j \hat\otimes 1)\tilde{P}_\mu \otimes \gamma^j$ 
on the projective module 
$\tilde{P}_\mu (E^{\oplus N}\otimes_A \calH_b \hat\otimes {C\ell_{0,1}}) \hat\otimes \bigwedge\nolimits^{\!*} \R^d$. 
This operator has compact resolvent 
by analogous arguments to the proof of Proposition \ref{prop:real_bulk_spec_trip}.

Combining our results so far, we consider the unbounded tuple
$$ 
\left(C\ell_{n,1}\hat\otimes C\ell_{0,d},\, \tilde{P}_\mu (E^{\oplus N} \otimes_A \calH_b \hat\otimes C\ell_{0,1})_{C\ell_{0,1}} \hat\otimes \bigwedge\nolimits^{\!*} \R^d,\, \sum_{j=1}^d \tilde{P}_\mu (1\otimes_\nabla X_j\hat\otimes 1)\tilde{P}_\mu \hat\otimes\gamma^j \right)
$$
with grading $(\mathrm{Ad}_{R}\hat\otimes 1\hat\otimes \gamma_{C\ell_{0,1}})\hat\otimes \gamma_{\bigwedge^* \R^d}$.
By construction, all Clifford generators have odd grading and graded-commute with the 
Dirac-type operator, which also commutes with the right $C\ell_{0,1}$ action. 
Hence the tuple is an unbounded Kasparov module. 
A simple check shows that the Kasparov module satisfies Kucerovsky's 
criterion~\cite[Theorem 13]{Kucerovsky97} and so is an unbounded representative of the product.
\end{proof}

We now have an unbounded representative of the product 
$[H^G]\hat\otimes_A[\lambda]\in KKO(C\ell_{n,0}\hat\otimes C\ell_{0,d}, \R)$ or 
$KKO(C\ell_{n,1}\hat\otimes C\ell_{0,d}, C\ell_{0,1})$. Our task is to 
associate an index to this class.

Let's begin with the simpler case of $G\subset\{1,T\}$ and so 
the product $[H^G]\hat\otimes_A [\lambda]$ is represented by a real spectral triple
$$
\left(C\ell_{n,0}\hat\otimes C\ell_{0,d},\, P_\mu (E^{\oplus N} \otimes_A \calH_b ) \hat\otimes \bigwedge\nolimits^{\!*} \R^d,\, \sum_{j=1}^d P_\mu (1\otimes_\nabla X_j)P_\mu \hat\otimes\gamma^j,\, (\Gamma\hat\otimes 1)\hat\otimes \gamma_{\bigwedge^* \R^d} \right).
$$
We let $P\wt{X}P$ denote the product operator.
Representing the $\Z_2$-grading as $(\begin{smallmatrix} 1 & 0\\ 0 & -1\end{smallmatrix})$, we can express 
$P\wt{X}P=\begin{pmatrix} 0 & P\wt{X}_-P \\ P\wt{X}_+P & 0 \end{pmatrix}$, 
where $\wt{X}_\pm$ are real Fredholm operators. 
The operator $\wt{X}$ graded-commutes with a 
left action of $C\ell_{n,0}\hat\otimes C\ell_{0,d} \cong C\ell_{n,d}$. 
As $P\wt{X}P$ is Fredholm, 
$\Ker(P\wt{X}P)\cong \Ker(P\wt{X}P)^0\oplus \Ker(P\wt{X}P)^1$ is a finite-dimensional $\Z_2$-graded 
$C\ell_{n,d}$-module. Furthermore, $\Ker(P\wt{X}P)^0 \cong \Ker(P\wt{X}_+P)$.
\begin{defn}[\cite{ABS64}]
Denote by $\hat{\mathfrak{M}}_{r,s}$ the Grothendieck group of equivalence classes of real $\Z_2$-graded modules with an irreducible graded left-representation of $C\ell_{r,s}$.
\end{defn}

The subspace $\Ker(P\wt{X}P)$ represents a class in the quotient group $\hat{\mathfrak{M}}_{n,d}/i^*\hat{\mathfrak{M}}_{n+1,d}$, where $i^*$ comes from restricting a Clifford action of $C\ell_{n+1,d}$ to $C\ell_{n,d}$. Next, we use the Atiyah-Bott-Shapiro isomorphism~\cite[Theorem I.9.27]{SpinGeometry} to relate
$$  \hat{\mathfrak{M}}_{n,d}/i^*\hat{\mathfrak{M}}_{n+1,d} \cong KO^{d-n}(\mathrm{pt}) \cong KO_{n-d}(\R). $$

\begin{defn} \label{def:Clifford_index}
The Clifford index of $P\wt{X}P$ is given by
$$ \Index_{n-d}(P\wt{X}P) := [\Ker(P\wt{X}P)] \in \hat{\mathfrak{M}}_{n,d}/i^*\hat{\mathfrak{M}}_{n+1,d} \cong KO_{n-d}(\R). $$
\end{defn}

We remark that $\Index_k$ is a generalisation of the usual index. To see this, we 
first note that $C\ell_{0,0}\cong \R$ and $C\ell_{1,0}\cong \R\oplus\R$. A 
$\Z_2$-graded $C\ell_{0,0}$-module is given by any $\Z_2$-graded finite-dimensional 
real vector space $V^0\oplus V^1$. Next observe that $V\oplus V \cong V\otimes \R^2$ 
extends to a graded $C\ell_{1,0}$-module, which implies that $[V\oplus 0] = -[0\oplus V]$ 
in $\hat{\mathfrak{M}}_{0,0}/i^*\hat{\mathfrak{M}}_{1,0}$. Hence, given a Dirac-type 
operator $D$ such that $\Ker(D)$ is a $\Z_2$-graded $C\ell_{0,0}$-module,
\begin{align*}
  \Index_0(D) &= [\Ker(D)^0 \oplus \Ker(D)^1] \cong [\Ker(D)^0 \oplus 0] -[\Ker(D)^1 \oplus 0]  \\
          &\cong \mathrm{dim}_\R \Ker(D_+) - \mathrm{dim}_\R \,\mathrm{CoKer}(D_+)\in\Z \cong KO_0(\R).
\end{align*}
Therefore we see that $\Index_k$ reduces to the usual Fredholm index when $k=0$. 
We direct the reader to~\cite{AS69} and~\cite[Chapter \Rmnum{1}.9, \Rmnum{2}.7, \Rmnum{3}.10]{SpinGeometry} 
for more details on the Clifford index. 
A similar viewpoint on expressing the invariants in $KO_{n-d}(\R)$ as index-like 
maps is considered in~\cite{DNSB14b, GSB15}.

As a final detail, we must show that the spectral triple representing the product from 
Lemma \ref{lemma:real_index_pairing_product} does not contribute any topological information 
outside of $\Ker(P\wt{X}P)$. We relegate this detail to Appendix \ref{sec:Prod_mod_degenerate_off_kernel}.

Let us now consider the case of odd symmetries and 
$[H^G]\hat\otimes_A[\lambda] \in KKO(C\ell_{n,1}\hat\otimes C\ell_{0,d},C\ell_{0,1})$. 
Because we are now solely interested in the topological information of the product 
Kasparov module, we can use stability of the $KKO$ groups to associate $[H^G]\hat\otimes_A[\lambda]$ 
to a class
$$
   \left[ ( C\ell_{n,d}, \calH_\R, \hat{X}, \gamma ) \right] \in KKO(C\ell_{n,d},\R),
$$
which is now the class of a real spectral triple. Furthermore we can suppose without issue 
that $\hat{X}$ graded-commutes with the left $C\ell_{n,d}$ representation 
(see~\cite[Chapter 8]{HigsonRoe}). Thus we may do the same procedure as the 
case of $G\subset\{1,T\}$ and 
associate the Clifford index, $\Index_{n-d}(\hat{X}) \in KO_{n-d}(\R)$, to the 
product $[H^G]\hat\otimes_A [\lambda]$.


\subsection{Some brief remarks on disordered systems} \label{subsec:remarks_on_disorder}

Following Bellissard and others, \cite{Bellissard94}, if we wish to consider systems 
with disorder or impurities, our algebra of interest is the twisted
crossed product $C(\Omega)\rtimes_\phi \Z^d$ (if the algebra is real, 
we take the twist $\phi=0$). The compact space $\Omega$ is the disorder space 
of configurations and carries a probability measure such that the action $\{T_\alpha:\alpha\in\Z^d\}$ 
on $\Omega$ is invariant and ergodic.
The disordered
Hamiltonian $H_\omega$ is indexed by 
$\omega\in\Omega$ such that $H_\omega\in\pi_\omega(C(\Omega)\rtimes_\phi \Z^d)$ 
and $\{\pi_\omega\}_{\omega\in\Omega}$ is a family of representations linked by
the covariance relation 
$\wh{S}^\alpha \pi_\omega(b) \wh{S}^{-\alpha} = \pi_{T_\alpha \omega}(b)$ 
for $\wh{S}^\alpha = \wh{S}_1^{\alpha_1}\cdots \wh{S}_d^{\alpha_d}$ 
a (possibly twisted) translation operator and $b\in C(\Omega)\rtimes_\phi\Z^d$~\cite{Bellissard94}.

If for every $\omega\in\Omega$, $H_\omega$ still has a 
spectral gap at $0$, then $\mathrm{sgn}(H_\omega)$ will still give a
grading for the PUA representation. Hence our general method to 
obtain the Kasparov module encoding the internal 
symmetries of the Hamiltonian will extend using the algebra
$\pi_\omega(C(\Omega)\rtimes_\phi \Z^d)$. The 
Kasparov module gives a class in $KKO(C\ell_{n,0},\pi_\omega(C(\Omega)\rtimes_\phi \Z^d))$
and the covariance relation can then be used to ensure that the 
$KKO$-class is independent of the choice of $\omega\in\Omega$ 
(provided $[\wh{S}^\alpha,R_g]=0$ for all $g\in G$ and $\alpha\in\Z^d$, which
is true in all relevant examples).

Similarly, it is a simple extension of our existing proofs to show that
$$ 
\lambda_\omega= 
\left( \pi_\omega(\calA)\hat\otimes C\ell_{0,d}, \calH_b \otimes \bigwedge\nolimits^{\!*} \R^d,\,
\sum_{j=1}^d X_j \otimes  \gamma^j,\, \gamma_{\bigwedge^* \R^d} \right)  
$$
is a real spectral triple with $\calA = C_c(\Z^d, C(\Omega))$, a dense $\ast$-subalgebra. 
Similarly the spectral triple may include disorder in the complex case
in the same way. 

We compare different representations of the 
disorder parameter by the covariance relation, which gives the
unitarily equivalent spectral triple
$$  
\wh{S}^\alpha\lambda_\omega\wh{S}^{-\alpha} 
= \left( \pi_{T_\alpha\omega}(\calA)\hat\otimes C\ell_{0,d}, 
\calH_b \otimes\bigwedge\nolimits^{\!*} \R^d,\, 
\sum_{j=1}^d (X_j-\alpha_j) \otimes  \gamma^j,\, \gamma_{\bigwedge^* \R^d} \right) 
$$
as $\wh{S}^\alpha X_j \wh{S}^{-\alpha} = X_j - \alpha_j$. 
The straight line homotopy $\calD_t = \sum_j (X_j-t\alpha_j)\otimes \gamma^j$ for $t\in[0,1]$
shows that $[\lambda_\omega] = [\lambda_{T_\alpha\omega}]$ 
and the $KKO$-classes are orbit equivalent. 
As the action of $\Z^d$ on $\Omega$ is taken to be ergodic, 
all relevant Kasparov modules are independent of 
the choice of $\omega\in\Omega$. Hence the index
pairing $[(H^G)_\omega]\hat\otimes_{\pi_\omega(A)} [\lambda_\omega]$ 
is independent of $\omega$ and 
our topological invariants are stable under the addition of disorder.

As already stated, our disordered model currently requires the spectral gap of the Hamiltonian to persist 
under the
addition of disorder. This is an unrealistic assumption. 
Instead an analysis must be conducted by considering gaps in the
extended state spectrum, and the phenomenon of
localisation. It is one of the key achievements of Bellissard et al.~\cite{Bellissard94} 
that the complex pairing extends to regions under 
Anderson localisation. Such an extension is also obtained for the complex 
pairings in~\cite{PLB13, PSB14}. 
We delay a full treatment of the significant problem of localisation 
in the real/torsion setting to another place.

\section{Applications} \label{sec:Applications}

\subsection{Insulator models}

\begin{example}[Time-reversal invariant $2D$ Insulators, Kane-Mele model] 
\label{example:bulk_KaneMele}
We take $d=2$ and the subgroup $G=\{1,T\}$. 
We are modelling particles with spin $s=1/2$ and 
so the time-reversal involution $R_T$ is such that 
$R_T^2 = (-1)^{2s} =-1$. The operator $R_T$ is 
anti-unitary so we will work in the category of real 
algebras and modules. The time-reversal operator 
acts on $\calH = \ell^2(\Z^2)\otimes \C^{2N}$ 
(considered where necessary as a real Hilbert space) by the matrix
$$  
R_T = \begin{pmatrix} 0_N & \calC \\ -\calC & 0_N \end{pmatrix}, 
$$
where $\calC$ is pointwise complex conjugation. A self-adjoint 
operator that is invariant under conjugation by $R_T$ takes the 
form $\begin{pmatrix}a & b \\ -\calC b\calC & \calC a \calC \end{pmatrix}$, 
where $a$ and $\calC a\calC$ are self-adjoint and $b^* = -\calC b\calC$. 
Following~\cite{KM05b, DNSB14b}, we take the Hamiltonian
$$ 
H_{KM} = \begin{pmatrix} h & g \\ g^* & \calC h \calC \end{pmatrix}, 
$$
where $h$ is a Haldane Hamiltonian (that is, Hamiltonian of shift 
operators acting on a honeycomb lattice), and $g$ is the 
Rashba coupling~\cite{KM05b}. We either require the Rashba 
coupling to be such that $g^* = -\calC g\calC$ or it is sufficiently 
small so we may take a homotopy of $H_{KM}$ to a Hamiltonian 
with $g=0$~\cite{SchulzBaldes13, DNSB14b}. Typically $h$ and $g$ 
are matrices of finite polynomials of the shift operators 
$S_j$ (see~\cite[Section 5]{ASV13}). Provided $h$ and $g$ 
are such that $\mu\notin\sigma(H_{KM})$, $H_{KM}$ 
satisfies Assumption \ref{def:bulk_Hamiltonian_condition}. 
Therefore we take the algebra 
$A = M_{2N}(C^*(\Z^2)) \subset \calB[\ell^2(\Z^2)\otimes\C^{2N}]$ as a 
real $C^*$-algebra. Because the Fermi projection commutes with $R_T$, we may 
consider the Kasparov module from Proposition \ref{prop:TR_projective_KK_class} 
given by
$$ 
  \left( C\ell_{4,0}, P_\mu E_A^{\oplus 2}, 0, \Gamma \right).
$$
The Kasparov module gives a class $[H^G]\in KKO(C\ell_{4,0},M_{2N}[C^*(\Z^2)])$, which can 
be simplified by stability of $KKO$-groups to $KKO(C\ell_{4,0},C^*(\Z^2))$.

We can also consider the dense subalgebra $\calA\subset A$
and apply Proposition \ref{prop:real_bulk_spec_trip} to obtain the real spectral triple
$$ 
\lambda= \left(\calA \hat\otimes C\ell_{0,2}, \ell^2(\Z^2)\otimes\C^{2N}\otimes \bigwedge\nolimits^{\!*} \R^2, 
X_1\otimes 1_{2N}\otimes\gamma^1 + X_2\otimes 1_{2N}\otimes \gamma^2,  
\gamma_{\bigwedge^*\R^2}\right) 
$$
The left Clifford action is generated by $\rho^1$ and $\rho^2$, 
whose representation is given by 
Equation \eqref{eq:left_right_real_clifford_actions_on_exterior_algebra}. 
We can use the isomorphism of linear spaces $\bigwedge^*\R^2 \cong M_2(\C)$ 
to write explicit generators for our Clifford actions as matrices, 
though the result is independent of the choice of generators. 
For example, under a suitable identification of $i$ as a $2\times 2$ 
matrix that squares to $-1$, we can choose Clifford generators 
$\gamma^j$ such that our Dirac-type operator is of the form
$$  
X = \begin{pmatrix} 0_{2N} & X_1\otimes 1_{2N} -iX_2\otimes 1_{2N} \\ X_1\otimes 1_{2N} + iX_2\otimes 1_{2N} & 0_{2N} \end{pmatrix}, 
$$
which is analogous to the well-known Dirac-type operator of the quantum Hall effect.

We can then compute the product $[P_\mu^G]\hat\otimes_A [\lambda]$, which by 
Lemma \ref{lemma:real_index_pairing_product} is represented by the real spectral 
triple
\begin{equation} \label{eq:KM_pairing_product_triple}
  \left( C\ell_{4,0}\hat\otimes C\ell_{0,2}, P_\mu( E \otimes_A \calH )\otimes \R^2 \hat\otimes \bigwedge\nolimits^{\!*} \R^2,
     \sum_{j=1}^2 P_\mu (1\otimes_\nabla X_j)P_\mu \otimes 1_2 \hat\otimes \gamma^j, \Gamma\hat\otimes \gamma_{\bigwedge^* \R^2} \right).
\end{equation}
With $E_A\cong A^M$ ($M=\dim C\ell_{4,0}$) we have the isomorphism
$$
P_\mu( E \otimes_A \calH )\otimes \R^2 \otimes \bigwedge\nolimits^{\!*} \R^2\to 
P_\mu\calH^M \otimes \R^2 \otimes \bigwedge\nolimits^{\!*} \R^2,
$$ 
and then on the right hand side the Dirac-type operator becomes 
$\sum_j P_\mu(1_M\otimes X_j)P_\mu \otimes 1_2\otimes \gamma^j$.

Following Section \ref{subsec:KKO_Clifford_index} we can associate the spectral triple of 
Equation \eqref{eq:KM_pairing_product_triple} to a graded $C\ell_{4,2}$-module. Thus 
we consider the map
\begin{align*}
  &KKO(C\ell_{4,0},C^*(\Z^2)) \times KKO(C^*(\Z^2)\hat\otimes C\ell_{0,2}, \R) 
  \to KO_2(\R)   \\
   &\qquad \left([P_\mu^G],[X] \right) \mapsto \Index_{4-2}(P_\mu XP_\mu) \in KO_{2}(\R) \cong \Z_2 
\end{align*}
and so we obtain the well-known $\Z_2$ invariant that arises in such systems. 
The derived $\Z_2$ invariant is non-trivial provided the spin-orbit 
coupling in $h$ is sufficiently large and the Rashba coupling 
$g$ is controlled (see~\cite{DNSB14b,KM05b,  KK15}). Using the Atiyah--Bott--Shapiro isomorphism and an 
explicit choice of Clifford generators, we can express the Clifford index concretely as
$$  
  \Index_2(P_\mu X P_\mu) = \mathrm{dim}_\mathbb{H}\Ker(P_\mu X P_\mu)\,\mathrm{mod}\, 2 
    = \mathrm{dim}_\C \Ker\!\left( P_\mu (X_1 \otimes 1_M + i X_2\otimes 1_M) P_\mu \right)
       \mathrm{mod}\,2. 
$$
\end{example}

\begin{example}[$3D$ Topological insulators] \label{ex:3D_bulk_insulators}
Let us now consider some $3$-dimensional examples. 
What we consider does not encompass every possible 
$3D$-system, but will hopefully give a better understanding 
of how we apply our general $K$-theoretic picture.

Consider the space $\calH = \ell^2(\Z^3)\otimes \C^{2N}$ 
and the symmetry operators
\begin{align}  \label{eq:3D_defn_of_R_T_and_R_P}
   &R_T = \begin{pmatrix} 0_N & \calC \\ -\calC & 0_N \end{pmatrix},  
   &&R_C = \begin{pmatrix} 0_N & i\calC \\ i\calC & 0_N \end{pmatrix}.
\end{align}
These operators correspond to an odd time-reversal involution ($R_T^2=-1$) 
and an even charge-conjugation involution ($R_C^2=1$). First, we consider operators of the form
\begin{equation} \label{eq:3D_Hamilt}
h = i\left(\sum_{j=1}^3 \sum_{k_j}^{\text{ finite}} \alpha_{k_j}\!\left(S_j^{k_j}\otimes 1_N - (S_j^*)^{k_j}\otimes 1_N\right) \right) 
\end{equation}
on $\ell^2(\Z^3)\otimes\C^N$ with $\alpha_{k_j}\in\R$ for all $k_j$. Using $h$ we define
$$ H_{3D} = \begin{pmatrix} 0 & h \\ h & 0 \end{pmatrix}. $$
Because $h=h^*$ and $i\calC hi\calC = -h$, one can check that 
such Hamiltonians are time-reversal and charge-conjugation 
symmetric for $R_T$ and $R_C$ given in Equation \eqref{eq:3D_defn_of_R_T_and_R_P}. We choose coefficients 
$\alpha_{k_j}$ such that $H_{3D}$ has a spectral gap at $0$. 
Then $H_{3D}$ satisfies Assumption \ref{def:bulk_Hamiltonian_condition} 
and so we can apply our general method. Because $H_{3D}$ is 
compatible with the full symmetry group $\{1,T,C,CT\}$ with 
$R_T^2=-1$ and $R_C^2=1$, Proposition \ref{prop:projective_KK_class_odd} 
and Table \ref{table:PT_classes_repn_table} imply that the class 
$[H^G_{3D}]\in KKO(C\ell_{3,0},C^*(\Z^3))$.

We can use Proposition \ref{prop:real_bulk_spec_trip} and the dense real subalgebra $\calA\subset C^*(\Z^3)$ of finite polynomials of shift operators to build the spectral triple
$$ 
\lambda_{3D} = \left( \calA \hat\otimes C\ell_{0,3},\, \ell^2(\Z^3)\otimes\C^{2N} \otimes \bigwedge\nolimits^{\!*} \R^3,\, \sum_{j=1}^3 X_j\otimes\gamma^j,\, \gamma_{\bigwedge^*\R^3} \right) 
$$
with left and right Clifford actions given by Equation \eqref{eq:left_right_real_clifford_actions_on_exterior_algebra}. The pairing 
$[H^G_{3D}]\hat\otimes_A [\lambda_{3D}]$ can be represented 
by a real spectral triple $[(C\ell_{3,3},\calH,\hat{X}, \gamma)]\in KKO(C\ell_{3,3},\R)$. 
Therefore the Clifford index $[\Ker(\hat{X})] \in  \hat{\mathfrak{M}}_{3,3}/i^*\hat{\mathfrak{M}}_{4,3}$ 
reduces to the usual index
$$ 
\Index_{0}(\hat{X}) = \mathrm{dim}_\R \Ker( \hat{X}_+ ) - \mathrm{dim}_\R \,\mathrm{CoKer}(\hat{X}_+) \in \Z. 
$$
Hence, in this example of $d=3$ with $R_T^2=-1$ and $R_C^2=1$, 
the invariant of interest is the usual integer-valued index, though 
now seen as a special case of a much broader framework.

We now consider a different $3$-dimensional Hamiltonian, defined by the matrix
$$
 \breve{H}_{3D} = \begin{pmatrix} h+\breve{h} & 0 \\ 0 & -h+\breve{h} \end{pmatrix},  \quad \quad \breve{h}=p(S_1,S_2,S_3),
$$
where $h$ is given in Equation \eqref{eq:3D_Hamilt} and 
$p$ is a finite polynomial with real coefficients such that 
$p(S_1,S_2,S_3)$ is self-adjoint. The new Hamiltonian has the 
property $R_T \breve{H}_{3D} R_T^* = \breve{H}_{3D}$, but is 
not charge-conjugation symmetric. Provided $\mu\notin\sigma(\breve{H}_{3D})$, 
we obtain a class $[\breve{H}^G_{3D}]\in KKO(C\ell_{4,0},C^*(\Z^3))$ 
by Proposition \ref{prop:TR_projective_KK_class} and 
Table \ref{table:PT_classes_repn_table}. We use the same spectral triple
$$ 
\lambda_{3D} 
= \left( \calA \hat\otimes C\ell_{0,3},\, \ell^2(\Z^3)\otimes\C^{2N} 
\otimes \bigwedge\nolimits^{\!*} \R^3,\, 
\sum_{j=1}^3 X_j\otimes\gamma^j,\, \gamma_{\bigwedge^*\R^3} \right) 
$$
and class $[\lambda_{3D}]\in KKO(C^*(\Z^3)\hat\otimes C\ell_{0,3}, \R)$, 
whose product with $[\breve{H}_{3D}^G]$ is such that
\begin{align*}
 [\breve{H}_{3D}^G]\hat\otimes_{C^*(\Z^3)} [\lambda_{3D}] 
 &\cong \Index_{4-3}(P\wt{X}P )\in KO_1(\R)\cong \Z_2.
\end{align*}
We can express this index concretely as
$$ \Index_{1}(P\wt{X}P ) \cong \mathrm{dim}_{\C}\Ker(P\wt{X}P)\,\mathrm{mod}\,2 
  = \mathrm{dim}_{\R}\Ker(P\wt{X}_+P)\,\mathrm{mod}\,2. $$

We emphasise that the spectral triples used in the different 
$3$-dimensional examples are the same (up to unitary equivalence) 
and so represent the same $KO$-homology class. What differentiates 
the invariants of interest in the two examples are the different 
classes represented by $[H^G]\in KKO(C\ell_{n,0},C^*(\Z^3))$ 
for changing $G$ and $n$. Hence the symmetries change but 
the Dirac type operator of the Brillouin zone $X = \sum_j X_j\otimes \gamma^j$ 
is the same (up to equivalence of $KO$-homology classes) in a 
fixed dimension. Such an occurrence also appears in~\cite{DNSB14b,GSB15}.
\end{example}

\subsection{The Kasparov product and the bulk-edge correspondence} \label{subsec:bulk_edge}

One of the main applications of the $KKO$-approach is to the 
bulk-edge correspondence for topological insulator systems.
We give a basic overview here, though the full details are 
lengthy and will be dealt with elsewhere~\cite{BKR}.

Following~\cite{KR06,SBKR02,KSB04b,MT15b}, we link 
bulk and edge systems by the (real) Pimsner-Voiculescu short exact sequence
$$  
0 \to B\otimes\calK[\ell^2(\N)] \to \calT \to C^*(\Z^d) \to 0, 
$$
where $\calT$ acts on $\ell^2(\Z^{d-1}\otimes\N)\otimes\C^N$ and 
$B$ acts on $\ell^2(\Z^{d-1})\otimes\C^N$ with $C^*(\Z^d)\cong B\rtimes \Z$ 
under the action $\alpha(b) = {S}_d^* b {S}_d$~\cite{PV80}. As $B$ 
acts on $\ell^2(\Z^{d-1})\otimes\C^N$, we think of its elements as 
observables concentrated at the boundary/edge of the sample. 
By results in~\cite[{\S}7]{Kasparov80}, we can associate to this 
short exact sequence a class 
$[\mathrm{Ext}]\in KKO(C^*(\Z^d)\hat\otimes C\ell_{0,1},B)$ 
(see also Equation \eqref{eq:Ext_KKO_equivalence}). 
In particular, we build an unbounded representative of the 
extension class using a method similar to~\cite[Section 2]{BCR14} and~\cite{RRS15}.

By taking a dense subalgebra $\calB$ of the `edge algebra' $B$, 
we can apply Proposition \ref{prop:real_bulk_spec_trip} to obtain 
an edge spectral triple $\lambda_e$ that gives a class 
$[\lambda_e]\in KKO(B\hat\otimes C\ell_{0,d-1},\R)$.

As we will show in~\cite{BKR}, under the intersection product
\begin{align*}
   KKO(C^*(\Z^d)\hat\otimes C\ell_{0,1},B) \times KKO(B\hat\otimes C\ell_{0,d-1},\R) 
   &\to KKO(C^*(\Z^d)\hat\otimes C\ell_{0,d},\R), \\
       [\mathrm{Ext}] \hat\otimes_B [\lambda_e] &= (-1)^{d-1}[\lambda_b], 
\end{align*}
where $[\lambda_b]$ is the class of the bulk spectral triple from 
Proposition \ref{prop:real_bulk_spec_trip} and for $d$ even $-[\lambda_b]$ 
denotes the inverse of $\lambda_b$ in $K$-homology. Taking the product with 
the symmetry $KK$-class $[H^G]$,
\begin{align*}
   C_{n,d} &= [H^G] \hat\otimes_A [\lambda_b] 
   = (-1)^{d-1} [H^G] \hat\otimes_A [\mathrm{Ext}]\hat\otimes_B [\lambda_e], 
\end{align*}
Therefore we can express the real index pairing as a map
\begin{align*}
 &KKO(C\ell_{n,0},A)\times KKO(A\hat\otimes C\ell_{0,1},B)\times KKO(B\hat\otimes C\ell_{0,d-1},\R) \to KKO(C\ell_{n,0}\hat\otimes C\ell_{0,d},\R).
\end{align*}
By the associativity of Kasparov product, this will either be a pairing
$$ 
KKO(C\ell_{n,0}, A) \times KKO(A \hat\otimes C\ell_{0,d},\R) 
\to KKO(C\ell_{n,0}\hat\otimes C\ell_{0,d},\R) \cong KO_{n-d}(\R), 
$$
the bulk invariant studied in Section \ref{subsec:spec_trip_and_pairings}, or 
$$  
KKO(C\ell_{n,0}\hat\otimes C\ell_{0,1},B) 
\times KKO(B \hat\otimes C\ell_{0,d-1},\R) \to  KO_{n-d}(\R), 
$$
an invariant that comes from the edge algebra $B$ 
of a system with boundary. The factorisation of the 
bulk spectral triple ensures that regardless of our 
choice of $K$-theory class we pair with, the result is the same for bulk and edge. 
In particular non-trivial bulk pairings imply non-trivial 
edge pairings and vice versa, a bulk-edge correspondence. 
We again refer to~\cite{BKR} for the details.


\appendix

\section{Kasparov theory for real algebras} \label{sec:RealKasTheory}
\subsection{Basic notions}
While there is an extensive range of applications of complex Kasparov theory there has been
much less use of the real version.  Thus the basics are not well-known and there is
a point to giving a brief introduction to $KK$-theory for real $C^*$-algebras. 
\begin{defn}
A real $C^*$-algebra $A$ is a real Banach $\ast$-algebra such that 
$\|a^*a\| = \|a\|^2$ and $a^*a+1\geq 0$ for all $a\in A$.
\end{defn}

\begin{remark}[A note on real vs Real $C^*$-algebras]
The real Gelfand-Naimark theorem says that commutative real 
$C^*$-algebras are isomorphic to algebras of the form 
$$  C_0(X)^\tau = \{ f\in C_0(X,\C)\,:\,f(x^\tau) = \ol{f(x)}\,\text{ for all }x\in X\}, $$
where $(X,\tau)$ is a locally compact Hausdorff space with involution $\tau$, see~\cite{AK48, Rosenberg15}. 

More generally, given a Real $C^*$-algebra $(A,\tau)$ 
(a complex $C^*$-algebra $A$ with anti-linear involution 
$\tau$ that preserves multiplication), the subalgebra of 
elements in $A$ invariant under $\tau$, 
$A^\tau =\{ a\in A\,:\, a^\tau = a \}$, is a real $C^*$-algebra. 
There is an equivalence of the category of Real $C^*$-algebras 
with the category of real $C^*$-algebras (see~\cite{LS10} for more 
detail on the relation between real and Real algebras).
\end{remark}

\begin{defn}
A real Hilbert $A$-module is a real linear space ${E}$ over $\R$ with right action by a real $C^*$-algebra $A$ and $A$-valued inner product $(\cdot\mid\cdot)_A$ such that such that for all $e,f,g\in{E}$, $\lambda,\rho\in\R$ and $a\in A$,
\begin{enumerate}
  \item $(\lambda e)\cdot a = \lambda(e\cdot a) = e\cdot(\lambda a)$,
  \item $(e|\lambda f + \rho g) = \lambda(e| f) + \rho(e|g)$,
  \item $(e|f\cdot a) = (e|f)\cdot a$,
  \item $(e|f) = (f|e)^{*}$ as a member of the $C^{*}$-algebra $A$,
  \item $(e|e)\geq 0$ as an element of $A$,
  \item $(e|e) = 0$ if and only if $e = 0$,
  \item The space is complete under the norm $\|e\|_{{E}}^2:= \|(e|e)\|_A$.
\end{enumerate}
\end{defn}

Many situations where  complex Hilbert $C^*$-modules arise have natural real analogues.

\begin{example} \label{ex:C*-module_of_algebra_over_itself}
Let $A$ be a $C^{*}$-algebra and consider $A_{A}$, the $C^{*}$-module of $A$ over itself defined by the relations
\begin{align*}
  &a\cdot b = ab,     &(a|b) = a^{*}b
\end{align*}
The only condition worth checking from the definition is that $(a|a)=0$ if and only if $a=0$. Using the $C^{*}$-norm condition, 
$$
  (a|a) = 0 \Leftrightarrow a^{*}a=0 \Leftrightarrow \|a^{*}a\| 
  = 0 \Leftrightarrow \|a\|^2 = 0 \Leftrightarrow a = 0.
$$
\end{example}

\begin{example}
Take $E\to X$ to be a real vector bundle over a 
locally compact Hausdorff space $X$. 
Provided that there exists a positive real-valued form 
$(\cdot\mid\cdot)$ on $E$, then we can define the real 
$C^*$-module $\Gamma_0(E)_{C_0(X,\R)}$ with right-action 
by multiplication and inner-product via $(\cdot\mid\cdot)$.
\end{example}

Much like the case of Hilbert spaces, we are interested in 
linear transformations between $C^{*}$-modules. 
Though there are many similarities between operators 
on $C^{*}$-modules and operators on Hilbert spaces, 
the adjoint operator $T^{*}$ of a given operator $T$ may not always be defined. 

We will denote by $\End_A(E)$ the adjointable endomorphisms on the 
Hilbert $C^{*}$-module $E_A$: i.e. those $T:\,E\to E$ 
for which there exists an adjoint $T^*:\,E\to E$. 
For $f,g\in E_A$, define the rank-$1$ endomorphism 
$\Theta_{f,g}h = f\cdot(g|h)$ for $h\in E_A$. We define 
$\End_{A}^{00}({E})$ to be the endomorphisms of finite rank and are given by the set 
$$   
\End_{A}^{00}({E}) = \text{span}_\R\{\Theta_{e,f}\,:\, e,f\in{E}\}.  
$$
The compact operators $\End_{A}^{0}({E})$ are defined by 
$$  
\End_{A}^{0}({E}) = \overline{\text{span}}_\R\{\Theta_{e,f}\,:\, e,f\in{E}\} 
= \overline{\End_{A}^{00}({E})}, 
$$
where the closure is taken via the norm of $A$.

\begin{example} \label{example:A_compacts_itself}
Take $A$ as $C^{*}$-module over itself. Then for $a,b,c\in A$, $\theta_{a,b}c = ab^{*}c$. Hence, if we take the closure of all linear combinations of $\theta_{a,b}$ for all $a,b\in A$, then we clearly get $A$ back. Thus, $\End_A^0(A) = A$.
\end{example}

We will use $A,B, C$ to denote real  $C^*$-algebras.
\begin{defn}
A real unbounded Kasparov module for an algebra $\calA$ is a 4-tuple $(\calA, {}_\pi{E}_B, \calD, \gamma)$ where $\calA$ is dense in $A$ and
\begin{enumerate}
  \item $\gamma\in\End_B(E)$ is a  $\Z_2$-grading of the  real right-$B$ $C^*$-module ${E}_B$, 
  \item $\pi$ is a graded real endomorphism $\pi:\calA\to \End_B({E})$,
  \item $\calD:\Dom(\calD)\subset E_B\to E_B$ is an odd unbounded operator with $(1+\calD^2)$ dense in $E_B$ and for all $a\in \calA$,
$$  [\calD,\pi(a)]_\pm \in \End_B({E}),  \qquad \qquad \qquad \pi(a)(1+\calD^2)^{-1/2}\in \End_B^0({E}), $$
where $[\cdot,\cdot]_\pm$ is the graded commutator.
\end{enumerate}
\end{defn}
We will write Kasparov modules as $(\calA, E_B, \calD)$ when the representation $\pi$ is unambiguous.

\begin{remark}[Real Kasparov modules]
Kasparov modules for Real $C^*$-algebras require extra structure. 
Namely an anti-linear involution on the space $E$ and algebra $B$ 
such that $e^{\tau_E} \cdot b^{\tau_B} = (e\cdot b)^{\tau_E}$ and 
that $(e^{\tau_E}\mid f^{\tau_E})_B = (e\mid f)_B^{\tau_B}$. One 
also requires the unbounded operator $\calD: \Dom(\calD)\to E_B$ to be 
invariant under the induced involution $\calD^\tau(e) = (\calD(e^\tau))^\tau$ 
for $e\in\Dom(\calD)$.

Because we do not require the extra structure that comes with $KKR$-theory 
in this paper, we will not investigate its properties further. We do however 
remark that Real algebras and modules appear to be the correct setting to 
study systems with parity/spatial involution symmetry. More generally, while
complexifications of real Kasparov modules yield Real Kasparov modules, 
the possibility of more complicated Real structures on a Real Kasparov 
module means that the two notions are distinct, \cite{AR}.
\end{remark}

The results of Baaj and Julg~\cite{BJ83} continue to hold for real 
Kasparov modules, so given an unbounded module 
$(\calA,{E}_B,\calD)$ we apply the Riesz map (bounded transform) 
$\calD\to \calD(1+\calD^2)^{-1/2}$ to obtain the real 
Kasparov module $(A,{E}_B,\calD(1+\calD^2)^{-1/2})$, 
where $A$ is the $C^*$-closure of the dense subalgebra $\calA$.

One can define notions of unitary equivalence, 
homotopy and degenerate Kasparov modules in the real 
setting (see~\cite[\S{4}]{Kasparov80}) just as in the complex case. Hence we can 
define the group $KKO(A,B)$ as the equivalence classes 
of real (bounded) Kasparov modules modulo the equivalence 
generated by these relations.

The generality of the constructions and proofs in~\cite{Kasparov80} 
mean that all the central results in complex $KK$-theory carry 
over into the real (and Real) setting. In particular, the intersection product
$$  KKO(A,B) \times KKO(B,C) \to KKO(A,C) $$
is still a well-defined map and other important properties such as stability 
$$  
KKO(A\hat\otimes\calK(\calH),B)\cong KKO(A,B)  
$$
continue to hold, where $\calK(\calH)$ is the algebra of 
compact operators on a separable real Hilbert space. 

When we consider the unbounded picture and the product, we note that 
Kucerovsky's theorem, which gives us checkable conditions as to 
whether an unbounded module represents the 
Kasparov product, continues to hold in $KKO$-setting \cite[Theorem 13]{Kucerovsky97}. 
While such a result is implicit in Kucerovsky and Kasparov's work, 
the modules and products we consider are simple enough that these 
technicalities will not play a role, and all computations are explicit.

\subsection{Higher-order groups}
Clifford algebras are used to define higher $KKO$-groups and encode periodicity. In the real setting, we define
$$   C\ell_{p,q} = \text{span}_\R\!\left\{\left.\gamma^1,\ldots,\gamma^p,\, \rho^1,\ldots,\rho^q\,\right\vert\,(\gamma^i)^2 = 1,\,(\gamma^i)^* = \gamma^i,\,(\rho^i)^2=-1,
\,(\rho^i)^*=-\rho^i \right\} $$
with generating elements mutually anti-commuting.

\begin{example}
Consider the real space $\R^{p,q}$ with basis $\{ e_1,\ldots,e_p,\,\epsilon_1,\ldots,\epsilon_q \}$ 
from which we construct the exterior algebra $\bigwedge^*\R^{p,q}$. 
We can define an action of $C\ell_{p,q}$ on $\bigwedge^*\R^{p,q}$ by 
Clifford multiplication. We define 
$\eta^j(\omega) = e_j\wedge \omega +\iota(e_j)\omega$ for $1\leq j\leq p$ 
and $\nu^j(\omega) = \epsilon_j\wedge\omega +\iota(\epsilon_j)\omega$ 
for $1\leq j\leq q$, where $\iota(v)$ denotes the contraction of a form along $v$. 
One readily checks that the $\eta^j$ and $\nu^j$ satisfy the requirements to be generators
of $C\ell_{p,q}$.

Similarly, given $\R^d$ we can construct $\bigwedge^*\R^d$ and define 
representations of $C\ell_{d,0}$ and $C\ell_{0,d}$ with the generators
\begin{align*}
  &\gamma^j(\omega) = e_j \wedge \omega + \iota(e_j)\omega,   
  && \rho^j(\omega) = e_j\wedge \omega - \iota(e_j)\omega
\end{align*}
respectively. We note the sign change that ensures $(\rho^j)^2 = -1$.
\end{example}

We define higher-order  $KKO$-groups by tensoring with real Clifford algebras. Kasparov defines
\begin{equation} \label{eq:KKO_higher_order_def}
  K_{p,q}K^{r,s}O(A,B) := KKO(A\hat\otimes C\ell_{p,q}, B\hat\otimes C\ell_{r,s}). 
\end{equation}
The definition from Equation \eqref{eq:KKO_higher_order_def} simplifies immediately with the following result.
\begin{thm}[\S5, Theorem 4 of~\cite{Kasparov80}] \label{thm:Kasparov_simplify_K_pK^q}
Given real algebras $A$ and $B$, then for a fixed difference $(p-q)-(r-s)$ the groups $K_{p,q}K^{r,s}O(A,B)$ are canonically isomorphic.
\end{thm}
\begin{proof}
We note that $C\ell_{n,n}\cong\End_\R(\bigwedge^*\R^n)$. 
Therefore $C\ell_{n,n}
\cong \calK(\calH)$ for $\calH = \bigwedge^*\R^n$ 
and by stability $KKO(A\hat\otimes C\ell_{r,s},B) 
\cong KK(A\hat\otimes C\ell_{r+1,s+1},B)$, since 
$C\ell_{r+1,s+1}\cong C\ell_{r,s}\hat\otimes C\ell_{1,1}$. Hence
$$ 
KKO(A\hat\otimes C\ell_{p,q}, B\hat\otimes C\ell_{r,s}) 
\cong KKO(A\hat\otimes C\ell_{p,q}\hat\otimes C\ell_{s,r},B) \cong KK(A\hat\otimes C\ell_{p+s,q+r},B). 
$$
Up to stable isomorphism, the algebra $C\ell_{p+s,q+r}$ depends solely on $(p+s)-(q+r) = (p-q)-(r-s)$.
\end{proof}

\begin{remark}
Theorem \ref{thm:Kasparov_simplify_K_pK^q} 
implies that it is sufficient to define higher 
$KKO$-groups by tensoring by real Clifford algebras 
of the form $C\ell_{n,0}$ or $C\ell_{0,n}$ 
(though cases like $C\ell_{r,s}$ may still arise in examples). 

By Theorem \ref{thm:Kasparov_simplify_K_pK^q}, we find that
$$  
KKO(A\hat\otimes C\ell_{n,0},B) \cong KKO(A, B\hat\otimes C\ell_{0,n}). 
$$
It is clear that in the real picture  the placement of a Clifford algebra 
on the left or right in the bivariant group $KKO(\cdot,\cdot)$ is important. 
Furthermore, there is a difference between the algebra $C\ell_{n,0}$ and 
$C\ell_{0,n}$ that does not occur in the complex theory.
\end{remark}

We now clarify the relation between real $KK$-groups and real $K$-theory.
\begin{prop}[\cite{Kasparov80}, \S{6} Theorem 3] 
\label{prop:real_Real_k_with_kasparov_equivalence}
For trivially graded, $\sigma$-unital real algebras $A$, 
we have the isomorphism $KKO(C\ell_{n,0},A) \cong KO_n(A)$.
\end{prop}

Proposition \ref{prop:real_Real_k_with_kasparov_equivalence} implies that if $A\cong C(X)$  for some compact Hausdorff space $X$, then the group $KKO(C\ell_{n,0},C(X))\cong KO^{-n}(X)$ (note the sign change that one can ignore in the complex setting) and so we are back in the setting of Atiyah's $KO$-theory for spaces (see, for example,~\cite{SpinGeometry} for more on topological $KO$-theory). The reader may also consult~\cite{BL15} for a useful characterisation of $KO_n(A)$ for real $C^*$-algebras $A$ in terms of unitaries and involutions.

Like the complex case, there is an equivalence between 
short exact sequences of real $C^*$-algebras and certain real Kasparov modules, where
\begin{align} \label{eq:Ext_KKO_equivalence}
  \mathrm{Ext}_\R(A,B) &\cong KKO(A\hat\otimes C\ell_{0,1},B) \cong KKO(A, B\hat\otimes C\ell_{1,0}) 
\end{align}
for real, nuclear and separable algebras $A$ and $B$~\cite[\S{7}]{Kasparov80}.

We also briefly consider Bott periodicity. Because $KK$-groups 
are stable and Clifford algebras encode an algebraic periodicity 
with $C\ell_{0,8}\cong C\ell_{8,0}\cong M_{16}(\R)$, it follows 
that $KKO(A\hat\otimes C\ell_{8,0},B) \cong KKO(A, B\hat\otimes C\ell_{8,0}) \cong KKO(A,B)$. 
We would like to relate the algebraic periodicity of the $KK$-groups 
to a topological periodicity. Kasparov defines the suspension of an 
algebra $A$ by $\Sigma A = C_0(\R) \otimes A$, where $\Sigma A$ has the $\Z_2$-grading coming from $A$
only. A complicated argument (involving the product) shows that
$$  
KKO(\Sigma^n A\hat\otimes C\ell_{n,0}, B) \cong KKO(A, B) \cong KKO(A,\Sigma^n B \hat\otimes C\ell_{n,0}), 
$$
which relates algebraic periodicity to the more familiar topological periodicity 
(see~\cite[\S 5]{Kasparov80} for a proof). In the real setting, one also has the 
`dual suspension' coming from the 
real algebra 
$$
C_0(i\R) = \left\{f\in C_0(\R,\C)\,:\,  \ol{f(x)} = f(-x)\right\}, 
$$ 
where we define $\wh{\Sigma}A = C_0(i\R)\otimes A$ (with $C_0(i\R)$ trivially 
graded). The two suspensions arise 
because the complex algebra $C_0(\R,\C)$ has two different Real involutions, 
namely pointwise complex conjugation and $f^\tau(x) = \ol{f(-x)}$. Taking the 
real subalgebras of $C_0(\R,\C)$ invariant under the Real involutions gives $C_0(\R)$ and 
$C_0(i\R)$ respectively.
One finds that
$$
  KKO(\wh{\Sigma}^n A\hat\otimes C\ell_{0,n}, B) \cong KKO(A, B) \cong KKO(A,\wh{\Sigma}^n B \hat\otimes C\ell_{0,n})
$$
or $KKO(\Sigma \wh{\Sigma} A, B) \cong KKO(A,B) \cong KKO(A, \Sigma \wh{\Sigma}B)$.
See~\cite{SchroderKTheory} for more information.


\section{Simplifying the product spectral triple} \label{sec:Prod_mod_degenerate_off_kernel}
Recall from Lemma \ref{lemma:real_index_pairing_product} that the real index pairing 
$[H^G]\hat\otimes_A [\lambda]$ for $G\subset\{1,T\}$ is represented by the unbounded module 
\begin{equation} \label{eq:product_spec_trip}
 \left(C\ell_{n,0}\hat\otimes C\ell_{0,d},\, P_\mu(E^{\oplus N} \otimes_A \calH_b) \otimes \bigwedge\nolimits^{\!*} \R^d,\, \sum_{j=1}^d P_\mu(1\otimes_\nabla X_j)P_\mu \otimes\gamma^j,\, (\Gamma \otimes 1)\hat\otimes \gamma_{\bigwedge^* \R^d} \right). 
\end{equation}
We let $P\wt{X}P = \sum_{j=1}^d P_\mu(1\otimes_\nabla X_j)P_\mu \otimes\gamma^j$ and $\calH = P_\mu(E^{\oplus N} \otimes_A \calH_b)$. Associated to the real spectral triple of Equation \eqref{eq:product_spec_trip} is the real Fredholm module
$$  \left( C\ell_{n,d}, \calH, \wt{F}, \gamma \right), $$
where $\wt{F} = P\wt{X}P(1+(P\wt{X}P)^2)^{-1/2}$~\cite{BJ83}. Because $P\wt{X}P$ is self-adjoint and graded-commutes with the Clifford action, so does $\wt{F}$. Hence $[\pi(c),\wt{F}]_{\pm} = \pi(c)(\wt{F}-\wt{F}^*)=0$ for any $c\in C\ell_{n,d}$. What stops the Fredholm module being degenerate is that $(1-\wt{F}^2)\in\calK(\calH)$ is not necessarily zero. 

We use the (real) polar decomposition of $\wt{F} = V|\wt{F}|$ from~\cite[Theorem 1.2.5]{LiRealOpAlgs} and note that $\Ker(V) = \Ker(\wt{F}) = \Ker(P\wt{X}P)$. Because $\Ker(V) = \Ker(\wt{F})$, we can take the operator homotopy $F_t = V|\wt{F}|^t$, $t\in[0,1]$ to obtain the Fredholm module $\left( C\ell_{n,d},\calH, V, \gamma \right)$, which represents a class in $KKO(C\ell_{n,d},\R)$. The partial isometry $V$ is a real Fredholm operator as $1_\calH - V^*V$ is a finite-rank projection and so $V$ has a pseudo-inverse. 

Finally we write $\left( C\ell_{n,d},\calH, V, \gamma \right)
=\left( C\ell_{n,d},\Ker(P\wt{X}P), 0, \gamma \right)\oplus \left( C\ell_{n,d},V^*V\calH, V, \gamma \right)$, 
and the second summand is degenerate. Thus the 
$KK$-class and so the index depends only on the former module.

Summing up our discussion, the topological properties of the 
product spectral triple of Equation \eqref{eq:product_spec_trip} 
are wholly contained in the real Fredholm index of $V$ and, 
hence, are determined by $\Ker(V) = \Ker(P\wt{X}P)$. 
Therefore it suffices to consider the Clifford module 
properties of $\Ker(P\wt{X}P)$ as we have done in Section \ref{subsec:KKO_Clifford_index}.


\begin{thebibliography}{9}
  
\bibitem{AZ}
A. Altland and M. R. Zirnbauer. Nonstandard symmetry classes in mesoscopic normal-superconducting hybrid structures. \emph{Phys. Rev. B}, \textbf{55}(2)1142, 1997.
  
\bibitem{AK48}
R. F. Arens and I. Kaplansky. Topological representation of algebras. \emph{Trans. Amer. Math. Soc.}, \textbf{63}:457--481,
1948.

\bibitem{AR} M. F. Atiyah. $K$-theory and reality. \emph{Quart. J. Math.}, {\bf 2}:367--386, 1966.

\bibitem{ABS64}
  M. F. Atiyah, R. Bott and A. Shapiro. Clifford modules. \emph{Topology}, {\bf 3}, Supplement \textbf{1}(0):3--38, 1964.

\bibitem{AS69}
  M. F. Atiyah and I. M. Singer. Index theory for skew-adjoint Fredholm operators. \emph{Inst. Hautes \'{E}tudes Sci. Publ. Math.}, 
  \textbf{37}:5--26, 1969.

\bibitem{ASV13}
  J. C. Avila, H. Schulz-Baldes and C. Villegas-Blas. Topological invariants of edge states for periodic two-dimensional models. \emph{Math. Phys. Anal. Geom.}, \textbf{16}(2):137--170, 2013.
  
  

\bibitem{BJ83}
  S. Baaj and P. Julg. {Th\'{e}orie bivariante de Kasparov et op\'{e}rateurs non born\'{e}s dans les $C^*$-modules hilbertiens}. \emph{C. R. Acad. Sci. Paris S\'{e}r. \Rmnum{1} Math}, \textbf{296}(21):875--878, 1983.

\bibitem{Bellissard94}
   J. Bellissard, A. van Elst and H. {Schulz-Baldes}. {The noncommutative geometry of the quantum Hall effect}.  \emph{J. Math. Phys.}, \textbf{35}(10):5373--5451, 1994.

\bibitem{BHZ06}
B. A. Bernevig, T. L. Hughes and S. Zhang. Quantum spin hall effect and topological phase transition in HgTe
quantum wells. \emph{Science}, \textbf{314}(5806):1757--1761, 2006.

\bibitem{Blackadar}
  B. Blackadar. \emph{$K$-Theory for Operator Algebras}, volume 5 of \emph{Mathematical Sciences Research Institute Publications}. Cambridge Univ. Press, 1998.

\bibitem{BL15}
J. L. Boersema and T. A. Loring. $K$-Theory for real $C^*$-algebras via unitary elements with symmetries. arXiv:1504.03284, 2015.

\bibitem{BLR12}
J. L. Boersema, T. A. Loring and E. Ruiz. Pictures of $KK$-theory for real $C^*$-algebras and almost commuting matrices. \emph{Banach J. Math. Anal.}, \textbf{10}(1):27--47, 2016.


\bibitem{Bourne thesis}
C. Bourne. Topological states of matter and noncommutative geometry. PhD Thesis, Australian National Univerity, 2015.
  
\bibitem{BCR14}
  C. Bourne, A. L. Carey and A. Rennie. The bulk-edge correspondence for the quantum Hall effect in Kasparov theory. \emph{Lett. Math. Phys.}, \textbf{105}(9):1253--1273, 2015. \url{http://dx.doi.org/10.1007/s11005-015-0781-y}
  
\bibitem{BKR}
C. Bourne, J. Kellendonk and A. Rennie. The $K$-theoretic bulk-edge correspondence for topological insulators. To appear.

\bibitem{BMvS13}
 S. Brain, B. Mesland and W. D. van Suijlekom. {Gauge theory for spectral triples and the 
 unbounded Kasparov product}. arXiv:1306.1951, 2013. To appear in {J. Noncommut. Geom.}
 
 
 \bibitem{Connes85}
A. Connes. Non-commutative differential geometry. \emph{Inst. Hautes \'{E}tudes Sci. Publ. Math.}, \textbf{62}:41--144, 1985.
 
 \bibitem{Connes95}
A. Connes.
{Gravity coupled with matter and foundation of noncommutative
geometry}.
\emph{Comm. Math. Phys.} \textbf{182}:155--176, 1996.


\bibitem{DNG14a}
G. De Nittis and K. Gomi. Classification of ``Real'' Bloch-bundles: topological quantum systems of type AI. \emph{J. Geom. Phys.}, \textbf{86}(0):303--338, 2014.

\bibitem{DNG14b}
G. De Nittis and K. Gomi. Classification of ``Quaternionic'' Bloch-bundles: Topological insulators of type AII. \emph{Comm. Math. Phys.} \textbf{339}(1):1--55, 2015.

\bibitem{DNG15a}
G. De Nittis and K. Gomi. Differential geometric invariants for time-reversal symmetric Bloch-bundles: the
``Real'' case. arXiv:1502.01232, 2015.

\bibitem{DNG15b}
G. De Nittis and K. Gomi. Chiral vector bundles: A geometric model for class AIII topological quantum
systems. arXiv:1504.04863, 2015.

\bibitem{DNSB14b}
 G. De Nittis and H. Schulz-Baldes. Spectral flows associated to flux tubes. \emph{Ann. Henri Poincar\'{e}}, 2014. \url{http://dx.doi.org/10.1007/s00023-014-0394-5}



\bibitem{FR15}
I. Forsyth and A. Rennie. Factorisation of equivariant spectral triples in unbounded $KK$-theory. arXiv:1505.02863, 2015.


\bibitem{FM13}
D. S. Freed and G. W. Moore. Twisted equivariant matter. \emph{Ann. Henri Poincar\'{e}}, \textbf{14}(8):1927--2023, 2013.


\bibitem{ElementsNCG}
J. M. Gracia-Bond\'{i}a, J. C. V\'{a}rilly and H. Figueroa. \emph{Elements of Noncommutative Geometry}. 
Birkh\"{a}user Advanced Texts Basler Lehrb\"{u}cher. Birkh\"{a}user, Boston, 2001.

\bibitem{GP13}
G. M. Graf and M. Porta. Bulk-edge correspondence for two-dimensional topological insulators. \emph{Comm. Math. Phys.}, \textbf{324}(3):851--895, 2013.

\bibitem{GSB15}
J. Grossmann and H. Schulz-Baldes. Index pairings in presence of symmetries with applications to topological insulators. \emph{Comm. Math. Phys.}, 2015. \url{http://dx.doi.org/10.1007/s00220-015-2530-6}

\bibitem{HL10a}
M. B. Hastings and T. A. Loring. Almost commuting matrices, localized Wannier functions, and the quantum Hall effect. \emph{J. Math. Phys.}, \textbf{51}(1):015214, 2010.

\bibitem{HL11}
M. B. Hastings and T. A. Loring. Topological insulators and $C^*$-algebras: Theory and numerical practice. \emph{Ann. Physics}, \textbf{326}(7):1699--1759, 2011.

\bibitem{HigsonRoe}
N. Higson and J. Roe. \emph{Analytic K-Homology.} Oxford Mathematical Monographs. Oxford University Press, USA, 2001.

\bibitem{3d_verification}
D. Hsieh, D. Qian, L. Wray, Y. Xia, Y. S. Hor, R. J. Cava and M. Z. Hasan. A topological Dirac insulator in a quantum spin Hall phase. \emph{Nature}, \textbf{452}:970--974, 2008.

\bibitem{KL13}
  J. Kaad and M. Lesch. {Spectral flow and the unbounded Kasparov product}. 
  \emph{Advances in Mathematics}, \textbf{298}:495--530, 2013.


\bibitem{Karoubi}
M. Karoubi. \emph{$K$-Theory: An Introduction}. Classics in Mathematics. Springer, 2008.


    
\bibitem{KM05b}  
C. L. Kane and E. J. Mele. $Z_2$ topological order and the quantum spin Hall effect. \emph{Phys. Rev. Lett.}, \textbf{95}:146802, 2005.
  
\bibitem{Kasparov80}
  G. G. Kasparov. {The operator $K$-functor and extensions of $C^*$-algebras}. 
  \emph{Math. USSR Izv.}, \textbf{16}:513--572, 1981.

\bibitem{KasparovNovikov}
G. G. Kasparov. Equivariant $KK$-theory and the novikov conjecture. \emph{Invent. Math.}, \textbf{91}(1):147--201, 1988.

\bibitem{KK15}
  H. Katsura and T. Koma. The $\Z_2$ index of disordered topological insulators with time reversal symmetry. \emph{J. Math. Phys.}, \textbf{57}:021903, 2016.
  
\bibitem{Kellendonk15}
J. Kellendonk. On the $C^*$-algebraic approach to topological phases for insulators. arXiv:1509.06271, 2015.
  
\bibitem{KR06}
J. Kellendonk and S. Richard. Topological boundary maps in physics. In F. Boca, R. Purice and \c{S}. Str\u{a}til\u{a}, editors, \emph{Perspectives in operator algebras and mathematical physics}. Theta Ser. Adv. Math., volume 8, Theta, Bucharest, pages 105--121 2008, arXiv:math-ph/0605048.

\bibitem{SBKR00}
  J. Kellendonk, T. Richter and H. Schulz-Baldes. {Simultaneous quantization of edge and bulk Hall conductivity}. \emph{J. Phys. A}, \textbf{33}(2), 2000.
  
  
\bibitem{SBKR02}
  J. Kellendonk, T. Richter and H. Schulz-Baldes. {Edge current channels and Chern numbers in the integer quantum Hall effect}. \emph{Rev. Math. Phys.}, \textbf{14}(01):87--119, 2002.

  
\bibitem{KSB04a}
  J. Kellendonk and H. Schulz-Baldes. {Quantization of edge currents for continuous magnetic operators}. \emph{J. Funct. Anal.}, \textbf{209}(2):388--413, 2004.
  

\bibitem{KSB04b}
  J. Kellendonk and H. Schulz-Baldes. {Boundary maps for $C^*$-crossed products with with an 
  application to the quantum Hall effect}. \emph{Comm. Math. Phys.}, \textbf{249}(3):611--637, 2004.

\bibitem{KZ15}
R. Kennedy and M. R. Zirnbauer. Bott periodicity for $\Z_2$ symmetric ground states of gapped 
free-fermion systems. \emph{Comm. Math. Phys.}, 2015. \url{http://dx.doi.org/10.1007/s00220-015-2512-8}

\bibitem{Kitaev09}
A. Kitaev. Periodic table for topological insulators and superconductors. In V. Lebedev \& M. Feigel{'}man, editor, \emph{American Institute of Physics Conference Series}, volume 1134 of \emph{American Institute of Physics Conference Series}, pages 22--30, 2009.
 
 
\bibitem{KWBRBMQZ07}
M. K\"{o}nig, S. Wiedmann, C. Br\"{u}ne, A. Roth, H. Buhmann, L. W. Molenkamp, X. Qi and S. Zhang. Quantum spin Hall insulator state in HgTe quantum wells. \emph{Science}, \textbf{318}(5851):766--770, 2007.

\bibitem{Kubota15a}
Y. Kubota. Notes on twisted equivariant $K$-theory for $C^*$-algebras. arXiv:1511.05312, 2015.

\bibitem{Kubota15b}
Y. Kubota. Controlled topological phases and bulk-edge correspondence. arXiv:1511.05314, 2015. 

\bibitem{Kucerovsky97}
  D. Kucerovsky. {The $KK$-product of unbounded modules}. \emph{$K$-Theory}, \textbf{11}:17--34, 1997.

\bibitem{SpinGeometry}
H. B. Lawson and M. L. Michelsohn. \emph{Spin Geometry}. Princeton mathematical series. Princeton University Press, 1989.

\bibitem{LiRealOpAlgs}
B. Li. \emph{Real Operator Algebras}. World Scientific, New York, 2003.

\bibitem{LRV12}
S. Lord, A. Rennie and J. C. V\'{a}rilly. Riemannian manifolds in noncommutative geometry. \emph{J. Geom. Phys.}, \textbf{62}(7):1611--1638, 2012.

\bibitem{Loring15}
T. A. Loring. $K$-theory and pseudospectra for topological insulators. \emph{Ann. Physics}, \textbf{356}:383--416, 2015.

\bibitem{HL10b}
T. A. Loring and M. B. Hastings. Disordered topological insulators via $C^*$-algebras. \emph{EPL (Europhysics Letters)}, \textbf{92}(6):67004, 2010.

\bibitem{LS10}
T. A. Loring and A. P. W. S{\o}rensen. Almost commuting self-adjoint matrices — the real and self-dual cases. arXiv:1012.3494, 2010.

\bibitem{LS13}
T. A. Loring and A. P. W. S{\o}rensen. Almost commuting unitary matrices related to time reversal. \emph{Comm. Math. Phys.}, \textbf{323}(3):859--887, 2013.

\bibitem{LS14}
T. A. Loring and A. P. W. S{\o}rensen. Almost commuting orthogonal matrices. \emph{J. Math. Anal. Appl.}, \textbf{420}(2):1051--1068, 2014.

\bibitem{CM96}
  P. McCann and A. L. Carey. {A discrete model of the integer quantum Hall effect}. \emph{Publ. RIMS,
Kyoto Univ.} \textbf{32}:117--156, 1996.

\bibitem{spinQH_with_mag_field}
E. Y. Ma, M. R. Calvo, J. Wang, B. Lian, M. Muhlbauer, C. Brune, Y. Cui, K. Lai, W. Kundhikanjana, Y. Yang, M. Baenniger, M. K\"{o}nig, C. Ames, H. Buhmann, P. Leubner, L. W. Molenkamp, S. Zhang,  D. Goldhaber-Gordon, M. A. Kelly, Z. Shen. Unexpected edge conduction in mercury telluride quantum wells under broken time-reversal symmetry. \emph{Nat. Commun.}, \textbf{6}, 2015. \url{http://dx.doi.org/10.1038/ncomms8252}

\bibitem{MT15}
V. Mathai and G. C. Thiang. T-duality and topological insulators. \emph{J. Phys. A}, \textbf{48}(42):42FT02, 2015. 
\url{http://dx.doi.org/10.1088/1751-8113/48/42/42FT02}

\bibitem{MT15b}
  V. Mathai and G. C. Thiang. {T-duality simplifies bulk-boundary correspondence}. arXiv:1505.05250, 2015. 
  To appear in Comm. Math. Phys.
  
\bibitem{MT15c}
 V. Mathai and G. C. Thiang. T-duality trivializes bulk-boundary correspondence: some higher dimensional cases. arXiv:1506.04492, 2015.

\bibitem{MeslandMonster}
  B. Mesland. {Unbounded bivariant $K$-Theory and correspondences in noncommutative geometry}. \emph{J. Reine Angew. Math.}, \textbf{691}:101--172, 2014.
  
\bibitem{MeslandSurvey}
  B. Mesland. {Spectral triples and $KK$-theory: A survey}. \emph{Clay Mathematics Proceedings}, volume 16: \emph{Topics in noncommutative geometry}, pp 197--212, 2012.


\bibitem{MR15}
B. Mesland and A. Rennie. Nonunital spectral triples and metric completeness in unbounded $KK$-theory. arXiv:1502.04520, 2015.

\bibitem{NB90}
S. Nakamura and J. Bellissard. Low energy bands do not contribute to quantum Hall effect. \emph{Comm. Math. Phys.}, \textbf{131}:283--305, 1990.



\bibitem{PR06}
  D. Pask and A. Rennie. {The noncommutative geometry of graph $C^*$-algebras \Rmnum{1}: The index theorem}. \emph{J. Funct. Anal.}, \textbf{233}(1):92--134, 2006.


\bibitem{PV80}
  M. Pimsner and and D. Voiculescu. {Exact sequences for $K$-groups and Ext-groups of certain cross-product $C^*$-algebras}. \emph{J. Operator Theory}, \textbf{4}(1):93--118, 1980.
 

\bibitem{Prodan10}
E. Prodan. Non-commutative tools for topological insulators. \emph{New J. Phys.}, \textbf{12}(6):065003, 2010.

\bibitem{Prodan11}
E. Prodan. Disordered topological insulators: a non-commutative geometry perspective. \emph{J. Phys. A}, \textbf{44}(11):113001, 2011.

\bibitem{Prodan14}
E. Prodan. The non-commutative geometry of the complex classes of topological insulators. \emph{Topol. Quantum Matter}, \textbf{1}(1):1--16, 2014.


\bibitem{PLB13}
E. Prodan, B. Leung, and J. Bellissard. The non-commutative $n$th-Chern number ($n\geq 1$). \emph{J. Phys. A}, \textbf{46}:5202, 2013. 

\bibitem{PSB14}
E. Prodan and H. Schulz-Baldes. Non-commutative odd Chern numbers and topological phases of disordered
chiral systems. arXiv:1402.5002, 2014.

\bibitem{PSBbook}
E. Prodan and H. Schulz-Baldes. \emph{Bulk and Boundary Invariants for Complex Topological Insulators: From $K$-Theory to Physics}. Springer, 2016, arXiv:1510.08744.

\bibitem{QHZ08}
X. Qi, T. L. Hughes and S. Zhang. Topological field theory of time-reversal invariant insulators. \emph{Phys. Rev. B}, \textbf{78}:195424, 2008.

\bibitem{Raeburn88}
I. Raeburn. On crossed products and Takai duality. \emph{Proc. Edinburgh Math. Soc. (2)}, \textbf{31}:321--330, 1988.

\bibitem{RRS15}
  A. Rennie, D. Robertson and A. Sims, {The extension class and KMS states for Cuntz-Pimsner algebras of some bi-Hilbertian bimodules}. arXiv:1501.05363, 2015.

\bibitem{Rosenberg15}
J. Rosenberg. Structure and applications of real $C^*$-algebras. arXiv:1505.04091, 2015.

\bibitem{RSFL10}
S. Ryu, A. P. Schnyder, A. Furusaki and A. W. W. Ludwig. Topological insulators and superconductors:
tenfold way and dimensional hierarchy. \emph{New J. Phys.}, \textbf{12}(6):065010, 2010.

\bibitem{SRFL08}
A. P. Schnyder, S. Ryu, A. Furusaki and A. W. W. Ludwig. Classification of topological insulators and
superconductors in three spatial dimensions. \emph{Phys. Rev. B}, \textbf{78}:195125, 2008.

\bibitem{SchroderKTheory}
H. Schr\"oder. \emph{$K$-Theory for Real $C^*$-algebras and Applications}. Taylor \& Francis, 1993.

\bibitem{SchulzBaldes13}
H. Schulz-Baldes. Persistence of spin edge currents in disordered quantum spin Hall systems. \emph{Comm. Math. Phys.}, \textbf{324}(2):589--600, 2013.

\bibitem{SchulzBaldes13b}
H. Schulz-Baldes. $\mathbb{Z}_2$ indices of odd symmetric Fredholm operators. \emph{Documenta Math.} 
\textbf{20}:1481--1500, 2015.

\bibitem{SCR11} M. Stone, C. Chiu and A. Roy. Symmetries, dimensions and topological
insulators: the mechanism behind the face of the Bott clock. \emph{J. Phys. A}, \textbf{44}(4):045001, 2011.

\bibitem{Thiang14}
G. C. Thiang. On the K-theoretic classification of topological phases of matter. \emph{Ann. Henri Poincar\'{e}}, 
\textbf{17}(4):757--794, 2016.

\bibitem{Thiang14b}
G. C. Thiang. Topological phases: Isomorphism, homotopy and $K$-theory. \emph{Int. J. Geom. Methods Mod. Phys.}, 
\textbf{12}:1550098, 2015.

\bibitem{TKNN}
D. J. Thouless, M. Kohmoto, M. P. Nightingale and M. den Nijs. Quantized Hall conductance in a two-dimensional periodic potential. \emph{Phys. Rev. Lett.}, \textbf{49}:405--408, 1982.

\bibitem{Witten15}
E. Witten. Fermion path integrals and topological phases. arXiv:1508.04715, 2015.


\end{thebibliography}
\end{document}